\documentclass[11pt,leqno]{amsart}
\usepackage{amsfonts,amssymb,epsfig,latexsym}
\usepackage{amsmath,comment,bm,mathrsfs,braket,bbm}
\usepackage[latin1]{inputenc}
\usepackage{color,layout}
\usepackage{ifthen}
\usepackage{a4wide}
\usepackage{hyperref}
\hypersetup{%
  colorlinks = true,
  linkcolor  = blue,
  citecolor    = red
}

\usepackage{xparse}
\makeatletter
\RenewDocumentCommand{\title}{om}{%
   \IfNoValueTF{#1}
     {\gdef\shorttitle{Resonance expansion for quantum walks}}%
     {\gdef\shorttitle{#1}}%
   \gdef\@title{#2}%
}
\makeatother
\usepackage{graphicx}[dvipdfmx,hiresbb]
\usepackage{color}  
\usepackage{bmpsize} 

\newtheorem{theorem}{Theorem}
\newtheorem{lemma}{Lemma}[section]
\newtheorem{proposition}[lemma]{Proposition}
\newtheorem{definition}[lemma]{Definition}
\newtheorem{remark}[lemma]{Remark}
\newtheorem{corollary}[lemma]{Corollary}

\newtheorem*{assumption}{Assumption}

\def\square{\hbox{\vrule\vbox{\hrule\phantom{o}\hrule}\vrule}}

\graphicspath{{./Figures/}}

\newcommand{\be}{\begin{equation}}
\newcommand{\ee}{\end{equation}}
\newcommand{\ben}{\begin{equation*}}
\newcommand{\een}{\end{equation*}}

\newcommand{\til}[1]{\widetilde{#1}}

\numberwithin{equation}{section}

\newcommand{\N}{\mathbb{N}}
\newcommand{\Z}{\mathbb{Z}}
\newcommand{\R}{\mathbb{R}}
\newcommand{\C}{\mathbb{C}}

\newcommand{\cH}{{\mathcal H}}

\newcommand{\p}{\partial}
\newcommand{\e}{\varepsilon}
\newcommand{\dl}{\delta}
\newcommand{\pphi}{\varphi}

\newcommand{\im}{{\rm Im}\hskip 1pt }
\newcommand{\ord}{{\mathcal O}}

\newcommand{\ope}[1]{{\operatorname{#1}}}
\newcommand{\mc}[1]{{\mathcal{#1}}}

\newcommand{\qtext}[1]{\quad\text{#1 }\ }

\newcommand{\out}{{\sharp}}
\newcommand{\inc}{{\flat}}

\newcommand{\refup}[1]{$\ref{#1}$}

\numberwithin{equation}{section}

\begin{document}

\title{
Resonance expansion for quantum walks and its applications to the long-time behavior
}
\author{Kenta Higuchi}
\author{Hisashi Morioka}
\author{Etsuo Segawa}


\maketitle

\begin{abstract}
In this paper, resonances are introduced to a class of quantum walks on $\mathbb{Z}$. Resonances are defined as poles of the meromorphically extended resolvent of the unitary time evolution operator. In particular, they appear inside the unit circle.  Some analogous properties to those of quantum resonances for Schr\"odinger operators are shown. Especially, the resonance expansion, an analogue of the eigenfunction expansion, indicates the long-time behavior of quantum walks. The decaying rate, the asymptotic probability distribution, and the weak limit of the probability density are described by resonances and associated (generalized) resonant states. The generic simplicity of resonances is also investigated.

\end{abstract}

\section{Introduction}
Resonances, a generalization of eigenvalues, are known as characteristic quantities to observe the long-time behavior in various problems. The real and imaginary parts of a resonance are interpreted as the rate of oscillations and that of decay of a physical state, respectively. For example, in the study of Schr\"odinger equations, the imaginary part of each resonance gives the reciprocal of the half-life time of an associated state. 
In this manuscript, we introduce resonances to the discrete-time quantum walk on $\Z$. We observe some properties such as decaying rate and probability distribution of quantum walkers from the resonance expansion, which is similar to the eigenfunction expansion. We first briefly explain our results in the following three Subsections, and then mention backgrounds, motivations, and related works in Subsection~\ref{sec:Background}.

\subsection{Motivating example}\label{Sec:Moti-Exam}
Let us introduce a discrete-time quantum walk on $\Z$. For a given initial state $\psi\in\cH:=l^2(\Z;\C^2)$, the state at the time $n\in\Z$ is given by
\be
\psi_n:=U^n\psi.
\ee
Here, the unitary operator $U=SC$ is defined for a given sequence $(C(x))_{x\in\Z}$ of $2\times2$ unitary matrices as follows; 
\ben
(C\psi)(x)=C(x)\psi(x),\quad
(S\psi)(x)=
\begin{pmatrix}1&0\\0&0\end{pmatrix}\psi(x+1)+\begin{pmatrix}0&0\\0&1\end{pmatrix}\psi(x-1)\quad(x\in\Z,\ \psi\in\cH).
\een
We call $C$ and $S$ coin and shift, respectively. 

\begin{figure}
\centering
\includegraphics[bb=0 0 1129 221, width=15cm]{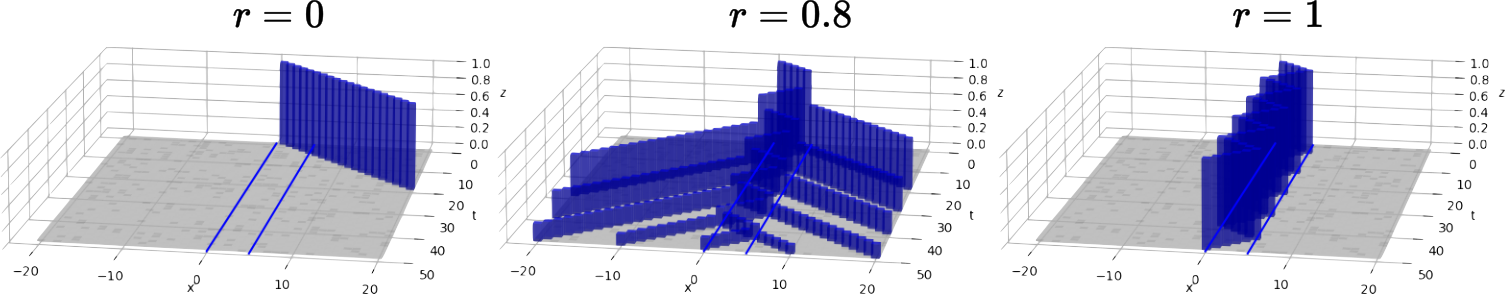}
\caption{\small Time evolution of $\|U_r^t\psi(x)\|_{\C^2}$ with $k=5$, $\psi(1)={}^t(0,1)$, $\psi(x)=0$ $(x\neq1)$}
\label{Fig:Evo}
\end{figure}
As a motivating example, let us start with the simplest case called double barrier problem (see Figure~\ref{Fig:Evo} and Proposition~\ref{prop:DoubleBarrier}). 
Let $k$ be a positive integer and let $0\le r\le1$. Define $U_r=SC_r$ by $C_r(x)=I_2$ (the identity matrix) for $x\in\Z\setminus\{0,k\}$ and 
\ben
C_r(0)=C_r(k)=\begin{pmatrix}\sqrt{1-r^2}&r\\-r&\sqrt{1-r^2}\end{pmatrix}.
\een
Suppose that $\psi(x)=0$ holds except for a finite number of $x\in\Z$. 
Then for the free quantum walk $U_0=S$, the first (resp. second) entry of each vector $\psi(x)\in\C^2$ shifts to left (resp. right) each time $U_0$ is applied, and we do not find any quantum walker in each compact set after a certain time passing. Contrary, for $U_1$, we find $2k$ eigenvalues $\zeta_1,\ldots,\zeta_{2k}$ $(\zeta_j=\exp(i\pi(2j-1)/2k))$ and $\psi_n$ varies periodically with its period $2k$ (with respect to $n\gg1$) in each compact set. Moreover, the behavior in the compact set is given by a linear combination of eigenstates.

For the intermediate cases $0<r<1$, the time evolution is similar to neither of the above cases. 
We may find quantum walker between $0$ and $k$ also for large $n$. However, there is no eigenvalue of $U_r$. Moreover, $\psi_n$ is $2k$-quasi periodic in each compact set. Let $J\subset\R$ be a large interval.  We have
\be
\psi_{n+2k}(x)=-r^2\psi_n(x)
\ee
for any $x\in J\cap\Z$ and $n\ge \exists n_0\in\N$. This behavior is explained in a similar way to that of $U_1$ by generalizing eigenvalues to resonances. In this case, $\lambda_1,\ldots,\lambda_{2k}$ given by
\be
\lambda_j=r^{\frac 1k}\zeta_j=r^{\frac 1k}e^{i\frac{\pi (2j-1)}{2k}}\quad(j=1,2,\ldots,2k)
\ee
are resonances (in the sense of Definition~\ref{def:Res}). Let $\pphi_j$ be a map $\Z\to\C^2$ defined by
\be
\pphi_j(x)=
\begin{pmatrix}
\hphantom{-r}(\mathbbm{1}_{(-\infty,-1]}(x)\sqrt{1-r^2}+\mathbbm{1}_{[0, k-1]}(x))\lambda_j^{x\hphantom{-}}\\
-r(\mathbbm{1}_{[k+1,+\infty)}(x)\sqrt{1-r^2}+\mathbbm{1}_{[1,k]}(x))\lambda_j^{-x}
\end{pmatrix},
\ee
with $\mathbbm{1}_A$ standing for the characteristic function of each subset $A$ of $\R$:
\ben
\mathbbm{1}_A(x)=\left\{\begin{aligned}&1&&\text{when }\ x\in A\\&0&&\text{when }\ x\notin A.\end{aligned}\right.
\een
Then $\pphi_j$ satisfies the eigenequation
\ben
U_r\pphi_j=\lambda_j\pphi_j.
\een
Since $\left\|\pphi_j(x)\right\|_{\C^2}$ grows exponentially as $\left|x\right|\to\infty$, $\pphi_j$ does not belong to $\cH$, and $\lambda_j$ is not an eigenvalue of $U_r$ (see Figure~\ref{Fig:Res_St}). However, in $J$, each initial state $\psi$ supported on $J$ is decomposed into a linear combination
\ben
\psi=\sum_{j=1}^{2k}c_j\pphi_j+\pphi_0,
\een
where $c_1,\ldots,c_{2k}$ are constants, and $\pphi_0\in\cH$ is such that $U^n \pphi_0=0$ for $n\gg1$. 
\begin{figure}
\centering
\includegraphics[bb=0 0 400 292, width=10cm]{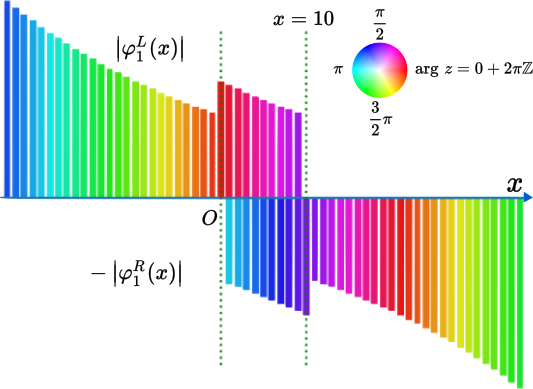}
\caption{Resonant state $\pphi_1$ with $r=2^{-1/2}$, $k=10$}
\label{Fig:Res_St}
\end{figure}
This implies that in $J$, we have
\ben
\psi_n=U_r^n\psi=\sum_{j=1}^{2k}c_j\lambda_j^n\pphi_j
\een
for large $n$. This is the resonance expansion of $\psi_n$ by resonant states of $U_r$ (Theorem~\ref{thm:ResExp} for general cases). In particular, this is quasi-periodic since $\lambda_j^{2k}=(r^{1/k}\zeta_j)^{2k}=r^2$. From this formula, we easily see the decaying rate of the probability to find a quantum walker in $J$:
\ben
\|\mathbbm{1}_J\psi_n\|_{\cH}
=r^{n/k}\left\|\mathbbm{1}_J\sum_{j=1}^{2k}c_j\zeta_j^n\pphi_j\right\|_{\cH}\le
r^{n/k}\left(\sum_{j=1}^{2k}\left|c_j\right|\|\mathbbm{1}_J\pphi_j\|_{\cH}\right)=M(\psi) r^{n/k},
\een
where $M(\psi)$ is a constant determined by $\psi$. The decaying rate is the modulus of the resonances: $r^{1/k}=\left|\lambda_j\right|$  (Corollary~\ref{cor:EstiSurvProb}).

\subsection{Definition of the resonance and resonance expansion}
Let us make precise and generalize the above argument. The resonances are defined as poles of the (meromorphically continued) resolvent operator of $U$. 
Throughout this manuscript, we consider finite rank perturbations without any ``isolations" on $\Z$:
\begin{assumption}
The diagonal entries of each $C(x)$ never vanishes for $x\in\Z$, and $C(x)=I_2$ except a finite number of $x\in\Z$. 
\end{assumption}
Since $U$ is unitary, the spectrum is a subset of the unit circle. Moreover, under Assumption, there is no eigenvalue of $U$, the spectrum of $U$ is the unit circle, and is absolutely continuous \cite[Lemmas 2.1 and 2.2]{MoSe}.  We note that the non-existence of eigenvalues is not necessary for studying resonances. For most of our results, one may find an analogue even if the diagonal entries of $C(x)$ vanishes for some $x\in\Z$.  However, the unique continuation principle for the equation $(U-\lambda)\psi=0$ follows from the condition (shown for example by the method of transfer matrices e.g. proof of Lemma~\ref{lem:Chara-T}), and it simplifies arguments.

Let $R(\lambda):=(U-\lambda)^{-1}$ for $\left|\lambda\right|>1$ be the resolvent operator (bounded on $\cH\to \cH$), and let $J\subset\R$ be a bounded interval.
As we will see in Proposition~\ref{prop:ContiResolv}, the cut-off resolvent $R_J(\lambda):=\mathbbm{1}_JR(\lambda)\mathbbm{1}_J$ extends meromorphically to whole $\lambda\in\C$. 
Moreover, the poles of $R_J$ and the multiplicity of each non-zero pole are invariant with respect to the choice of the interval $J$ containing the perturbed region $J\supset \ope{chs}(C-I_2)$. Here, we denote the convex hull of the support $\ope{supp}(C-I_2):=\{x\in\Z;\,C(x)\neq I_2\}$ of $C-I_2$ by
\ben
\ope{chs}(C-I_2):=[\inf\ope{supp}(C-I_2),\,\sup\ope{supp}(C-I_2)].
\een
Taking these facts into account, we define resonances in the following way.
\begin{definition}\label{def:Res}
We say $\lambda\in\C$ is a \textit{resonance} of $U$ if it is a pole of the extended family $\{R_J(\lambda);\,\lambda\in\C\}$, where $J$ is a bounded interval containing $\ope{chs}(C-I_2)$. We denote the set of resonances by $\ope{Res}(U)$. We define the (algebraic) multiplicity $m(\lambda)$ of each non-zero resonance $\lambda\in\ope{Res}(U)\setminus\{0\}$ by
\be
m(\lambda):=\ope{rank}\oint_\lambda R_J(\lambda')d\lambda',
\ee
where the integral runs over a small circle enclosing $\lambda$ counterclockwise. 
\end{definition}

As an analogue of the $L^2$-theory, we define the vector space of compactly supported states $\cH_{\ope{comp}}$ and that of locally $\cH$ maps $\cH_{\ope{loc}}$ by
\begin{align*}
&\cH_{\ope{comp}}:=\{\psi\in\cH;\,\ope{supp}\,\psi\text{ is compact}\},\\
&\cH_{\ope{loc}}:=\{\psi:\Z\to\C^2;\,\mathbbm{1}_J\psi\in\cH\text{ for any compact }J\subset\R\}.
\end{align*}
Note that $\cH_{\ope{loc}}$ coincides with the set of maps from $\Z$ to $\C^2$: $\cH_{\ope{loc}}=\{\psi:\Z\to\C^2\}$. 
We have seen in Subsection~\ref{Sec:Moti-Exam} that for the free quantum walk $U_0=S$, the motion of quantum walker is trivial. Let us denote the first and the second entry of a map $\psi\in\cH_{\ope{loc}}$ by $\psi^L$ and $\psi^R$, respectively. Then we have
\ben
(U_0\psi)(x)=\begin{pmatrix}\psi^L(x+1)\\\psi^R(x-1)\end{pmatrix}
\een
for each $x\in\Z$. We say $\psi$ is \textit{outgoing} if there exists $r>0$ such that
\ben
\psi^L(x)=\psi^R(-x)=0
\een
holds for any $x>r$. We also define the \textit{incoming support} of $\psi\in\cH_{\ope{loc}}$ by 
\ben
\ope{supp}^\inc \psi=\{x\in\Z;\,\inf\ope{supp}\,\psi^R\le x\le \sup\ope{supp}\,\psi^L\}.
\een
Then $\psi$ is outgoing if and only if $\ope{supp}^\inc \psi$ is compact.

The resolvent operator $(U-\lambda)^{-1}$ is characterized as the bounded operator which assigns to each $f\in\cH$ the solution $\psi\in\cH$ to the equation 
\be
(U-\lambda)\psi=f.
\ee
If $\lambda_0$ is an eigenvalue, a solution belonging to $\cH$ to the above equation is not unique since there is an eigenvector $\psi\in\cH$, that is, a non-trivial solution to $(U-\lambda_0)\psi=0$. Especially, $\lambda_0$ is a pole of the resolvent operator.  Similarly, the extended operator is characterized as the operator which assigns to $f\in\cH_{\ope{comp}}$ the outgoing solution of the above equation (Proposition~\ref{thm:DefRes} (\ref{enu:Reg})). A complex number $\lambda\in\C\setminus\{0\}$ is a resonance if and only if there exists an outgoing solution $\pphi_{\lambda}$ to the eigenequation (Proposition~\ref{thm:DefRes} (\ref{enu:Res-State})). We call $\pphi_\lambda$ a \textit{resonant state} associated with the resonance $\lambda$. Moreover, for each non-zero resonance $\lambda\in\ope{Res}(U)\setminus\{0\}$, there exists $\pphi_{\lambda,1},\ldots,\pphi_{\lambda,m(\lambda)}\in\cH_{\ope{loc}}$ which corresponds to the Jordan chain of an eigenvalue (Proposition~\ref{thm:DefRes} (\ref{enu:Jor-Chain})). They are also outgoing with $\mathcal{N}_1(\ope{supp}^\inc\pphi_{\lambda,k})\subset\ope{chs}(C-I_2)$, and we call each linear combination of them a \textit{generalized resonant state}. For a positive $r>0$ and a subset $A\subset\R$, we denote by 
\ben
\mathcal{N}_r(A):=\{x\in\R;\,\ope{dist}(x,A)\le r\},
\een
the closure of the $r$-neighborhood of $A$. 

Contrary, the pole at $\lambda=0$ is different from the others. The space 
\ben
V_J(0):=\ope{Ran}\oint_0 R_J(\lambda)\,d\lambda\subset\{\psi\in\cH_{\ope{comp}};\,\ope{supp}\,\psi\subset J\}
\een
depends on the choice of an interval $J$. For any $\pphi_0\in V_J(0)$, we have
\be
U^n\pphi_0(x)=0\qtext{on}J,
\ee
for any $n>2\left|J\right|_\Z$, where we put $\left|J\right|_\Z:=\ope{Card}(J\cap\Z)$.

\begin{theorem}\label{thm:ResExp}
For any compactly supported state $\psi\in\cH_{\ope{comp}}$ and any bounded interval $J\supset(\ope{supp}\,\psi\cup\ope{chs}(C-I_2))$, there exist coefficients $c_{\lambda,k}$ $(\lambda\in\ope{Res}(U)\setminus\{0\}$, $1\le k\le m(\lambda))$ and $\pphi_0\in V_J(0)$ such that
\be\label{eq:ResExp-Sum}
\psi=\mathbbm{1}_J\sum_{\lambda\in\ope{Res}(U)\setminus\{0\}}\sum_{k=1}^{m(\lambda)}c_{\lambda,k}\pphi_{\lambda,k}+\pphi_0.
\ee
Moreover, we have
\be\label{eq:ResExp-Time}
U^n\psi(x)=\sum_{\lambda\in\ope{Res}(U)\setminus\{0\}}\lambda^n\sum_{k=1}^{m(\lambda)}c_{\lambda,k}\sum_{l=0}^{k-1}\begin{pmatrix}n\\l\end{pmatrix}\lambda^{-l}\pphi_{\lambda,k-1}(x)\qtext{on}\mathcal{N}_{n-1-2\left|J\right|_\Z}(J),
\ee
for $n>2\left|J\right|_\Z$. Here, the (usual) binomial coefficient is defined for $n,l\in\N$ by
\ben
\begin{pmatrix}n\\l\end{pmatrix}=
\frac{n!}{l!(n-l)!}\quad\text{with }\ 0!=1.
\een
\end{theorem}

Since the number of resonances and their multiplicity is bounded by $2\left|\ope{chs}(C-I_2)\right|_{\Z}\ge 2|J|_\Z$ (Proposition~\ref{thm:DefRes} (\ref{enu:Sum-Multi})), the sums in \eqref{eq:ResExp-Sum} and \eqref{eq:ResExp-Time} are finite and the binomial coefficient is well-defined. As an analogue of the Schr\"odinger equation, these exact formulae may look strange. However, it is natural since the finiteness of the rank of the perturbation makes this problem essentially finite dimensional.

The latter formula \eqref{eq:ResExp-Time} is a consequence of the former with Lemma~\ref{lem:Outgoing-TimeEv}, which states that  if $\psi\in\cH_{\ope{loc}}$ is outgoing, then 
\be\label{eq:Commute-U1}
U^n\mathbbm{1}_J\psi=\mathbbm{1}_{\mathcal{N}_n(J)}U^n\psi
\ee
holds for any $n\ge0$ and for any interval $J$ containing $\ope{supp}^\inc\psi\cup\ope{chs}(C-I_2)$.

\begin{remark}
We will see in Theorem~\ref{thm:GenSimp} that in generic cases, all non-zero resonances of $U$ are simple. Then Formula~\ref{eq:ResExp-Time} turns into a simpler form
\ben
U^n\psi(x)=\sum_{\lambda\in\ope{Res}(U)\setminus\{0\}}\lambda^n c_{\lambda}\pphi_{\lambda}(x)\qtext{on}\mathcal{N}_{n-1-2\left|J\right|_\Z}(J).
\een
Most of the formulae in the next section follow from Formula~\ref{eq:ResExp-Time}, and each of them also turns into a simpler form by using the above formula.
\end{remark}

\begin{remark}
The existence of a non-identically vanishing outgoing solution $\pphi_\lambda$ to $(U-\lambda)\pphi=0$ for each non-zero resonance implies that the scattering matrix also has a pole there. In fact, we can equivalently define resonances as poles of the scattering matrix (Corollary~\ref{cor:Chara-S}).
\end{remark}

\subsection{Long-time behavior}\label{sec:LTB}
We observe some properties of long-time behavior of quantum walks by using the resonance expansion \eqref{eq:ResExp-Time} of Theorem~\ref{thm:ResExp}. 

The decaying rate of the survival probability on each compact interval is given in terms of the modulus of resonances and their multiplicity.
\begin{corollary}\label{cor:EstiSurvProb}
For any compactly supported initial state $\psi\in\cH_{\ope{comp}}$ and any bounded interval $J\supset (\ope{supp}\,\psi\cup\ope{chs}(C-I_2))$, there exists a constant $M=M(\psi)>0$ such that
\be\label{eq:DR1}
\|\mathbbm{1}_JU^n\psi\|_{\cH}\le Mn^{m_0-1}\Lambda_0
^n,
\ee
for $n$ large enough,
where $0\le\Lambda_0<1$ and $m_0\ge1$ stand for the maximal modulus of resonances 
and for the maximal multiplicity of non-zero resonance whose modulus attains $\Lambda_0$: 
\ben
\Lambda_0:=\max_{\lambda\in\ope{Res}(U)}\left|\lambda\right|,\quad
m_0:=p(\Lambda_0),\quad
p(\Lambda):=\max\left\{m(\lambda);\,\lambda\in\ope{Res}(U)\setminus\{0\},\ \left|\lambda\right|=\Lambda\right\}.
\een 
Moreover, there also exists $M'=M'(\psi)>0$ such that 
\be\label{eq:DR2}
\|\mathbbm{1}_JU^n\psi\|_{\cH}\le M'n^{p(\Lambda(\psi))-1}\Lambda(\psi)^n,
\ee
for $n$ large enough. 
We here put 
\be\label{eq:Lambda-psi}
\Lambda(\psi):=\max
\left\{\left|\lambda\right|;\,\oint_\lambda R(\lambda')\psi\,d\lambda'\neq0\right\}
\le\Lambda_0.
\ee 
\end{corollary}

Note that the decaying rate given by \eqref{eq:DR1} is independent of the choice of the initial state $\psi$ whereas that given by \eqref{eq:DR2} is not. In general, the latter is sharper. The sentence ``for $n$ large enough" means that there exists $n_0>0$ such that the statement is true for any $n\ge n_0$. The estimate \eqref{eq:DR1} (resp. \eqref{eq:DR2}) holds for any $n\in\N$ if $\Lambda_0\neq0$ (resp. $\Lambda(\psi)\neq0$).

For a given initial state $\psi\in\cH$ and $n\in\Z$, as usual, let $\mu_n:\Z^{\{0,1\}}\to[0,1]$ be a probability distribution defined by
\be
\mu_n(A):=\|\psi\|_{\cH}^{-2}\sum_{x\in A}\|U^n\psi(x)\|_{\C^2}^2\quad(A\subset\Z),
\ee
and let $X_n$ be the random variable induced by $\mu_n$. 
Since $U$ is unitary, $\mu_n(\Z)=\|\psi\|_{\cH}^{-2}\|U^n\psi\|_{\cH}^2=1$ holds for any $n$. 
The probability $\mu_n(x)$ that the position of a quantum walker at time $n$ be $x\in\Z$ is asymptotic to that given by resonant states.
\begin{corollary}\label{cor:PtWiseAsym}
For any compactly supported initial state $\psi\in\cH_{\ope{comp}}$ and any $x\in\Z$, we have
\be\label{eq:PtWiseAsym}
\mu_n(x)
=\Lambda(\psi)^{2n}\Biggl\|\sum_{
\substack{\lambda\in\ope{Res}(U)\\\left|\lambda\right|=\Lambda(\psi)}}e^{in\ope{arg}\lambda}\sum_{k=1}^{m(\lambda)}c_{\lambda,k}\sum_{l=0}^{k-1}\begin{pmatrix}n\\l\end{pmatrix}e^{-il\ope{arg}\lambda}\pphi_{\lambda,k}(x)\Biggr\|_{\C^2}^2
+o(\Lambda(\psi)^{2n})
\ee
as $n\to+\infty$, where $\Lambda=\Lambda(\psi)$ is given by \eqref{eq:Lambda-psi}. 
%
Especially, if the initial state is a restricted resonant state $\psi=\mathbbm{1}_J\pphi_\lambda$ with an interval $J\supset\ope{chs}(C-I_2)=:[x^-,x^+]$, there exists constants $c_\pm$ such that 
\ben
\mu_n(x)=
\left\{
\begin{aligned}
&c_\pm\left|\lambda\right|^{2(n\mp x)}\quad
&&\text{for }\ x\in \mathcal{N}_n(J)\setminus\ope{chs}(C-I_2)\qtext{with}\pm (x-x^\pm)>0,\\
&\left|\lambda\right|^n\|\pphi_\lambda\|_{\cH}^{-2}\|\pphi_\lambda(x)\|_{\C^2}^2&&\text{for }\ x\in\ope{chs}(C-I_2).
\end{aligned}\right.
\een
\end{corollary}
\begin{remark}
The estimate for the remainder $o(\Lambda(\psi)^{2n})$ in \eqref{eq:PtWiseAsym} is improved as follows:
\ben
\ord(n^{2p(\Lambda'(\psi))-2}\Lambda'(\psi)^{2n}),\quad
\Lambda'(\psi):=\max\left\{\left|\lambda\right|;\,\left|\lambda\right|<\Lambda(\psi),\ \oint_\lambda R(\lambda')\psi d\lambda'\neq0\right\}.
\een
\end{remark}

For an initial state $\psi\in\cH_{\ope{comp}}$ and an interval $J\supset(\ope{supp}\psi\cup\ope{chs}(C-I_2))$, the mean of the survival time in $J$ for quantum walkers is defined by
\ben
\tau=\tau(J,\psi):=\sum_{n=1}^\infty n\mu_n(\mc{N}_1(J)\setminus J).
\een
Corollary~\ref{cor:EstiSurvProb} shows the following upper bound of $\tau$ in terms of resonances.

\begin{corollary}\label{cor:Mean-time}
Suppose that there exists at least one non-zero resonance. Then we have the estimate
\ben
\tau(J,\psi)\le \min\left\{M(\psi)^2 2^{-2m_0+1}\Upsilon_{2m_0-1}(\Lambda_0),\ 
M'(\psi)^2 2^{-2p(\Lambda(\psi))+1}\Upsilon_{2p(\Lambda(\psi))-1}(\Lambda(\psi))\right\},
\een
where the function $\Upsilon_k$ for $k\in\N$ is defined by
\be\label{eq:Def-Upsilon}
\Upsilon_k(r):=r^{k-2}\frac{d^k}{dr^k}\left(\frac{1}{1-r^2}\right).
\ee
Here, we have used the same notations as in Corollary~\ref{cor:EstiSurvProb}. In particular, when $m_0=p(\Lambda)=1$, it turns into
\ben
\tau\le\min\left\{\left(\frac{M(\psi)}{1-\Lambda_0^2}\right)^2,\ 
\left(\frac{M'(\psi)}{1-\Lambda(\psi)^2}\right)^2\right\}.
\een
\end{corollary}
Remark that for $k=1,3,5$, the function $\Upsilon_k$ defined by \eqref{eq:Def-Upsilon} are given by
\ben
\Upsilon_1(r)=\frac2{(1-r^2)^2},\quad
\Upsilon_3(r)=\frac{24r^2(r^2+1)}{(1-r^2)^4},\quad
\Upsilon_5(r)=\frac{240r^4(3r^4+10r^2+3)}{(1-r^2)^6}.
\een
\begin{remark}
When the initial state $\psi=\mathbbm{1}_J\pphi_\lambda$ is a restriction of a resonant state $\pphi_\lambda$ associated with a non-zero resonance $\lambda\in\ope{Res}(U)\setminus\{0\}$, the mean $\tau$ is explicitly given by
\ben
\tau=
\frac1{1-\left|\lambda\right|^2}
\le\frac{1}{(1-\left|\lambda\right|^2)^2}=
\left(\frac{M(\psi)}{1-\Lambda_0^2}\right)^2=\left(\frac{M'(\psi)}{1-\Lambda(\psi)^2}\right)^2,
\een
where we have $M(\psi)=M'(\psi)=1$, $\Lambda_0=\Lambda(\psi)=\left|\lambda\right|$. 
\end{remark}

\begin{corollary}\label{cor:w-Lim}
For any compactly supported initial state $\psi\in\cH_{\ope{comp}}$ (suppose $\|\psi\|_\cH=1$ for simplicity), there exists the weak limit $W=\ope{w-}\lim_{n\to+\infty}(X_n/n)$. Its density function $w$ can be written in the form 
\be
w=c_-\dl_{-1}+c_+\dl_1,\quad c_\pm\ge0,\quad c_++c_-=1.
\ee
Furthermore,  for each $n\in\N$, we have
\be\label{eq:Conv-Spd}
\left|c_\pm-\left\|
\chi_\pm^\out U^n\psi\right\|_{\cH}^2\right|\le \left\|\chi^\inc U^n\psi\right\|_{\cH}^2=\ord(\Lambda(\psi)^{2n}),
\ee
where we put 
$\chi^\inc:=1-\chi_+^\out-\chi_-^\out$ with
\ben
\chi^\out_+(x)=\left\{
\begin{aligned}
&\begin{pmatrix}0&0\\0&1\end{pmatrix}&&x>\max\ope{chs}(C-I_2),\\
&0&&\text{otherwise},
\end{aligned}\right.\quad
\chi^\out_-(x)=\left\{
\begin{aligned}
&\begin{pmatrix}1&0\\0&0\end{pmatrix}&&x<\min\ope{chs}(C-I_2),\\
&0&&\text{otherwise}.
\end{aligned}\right.
\een
Especially, if $\psi=\mathbbm{1}_J\pphi_\lambda$ with a resonant state $\pphi_\lambda$ associated with a non-zero resonance $\lambda$ and an interval $J\supset(\ope{supp}\,\psi\cup\ope{chs}(C-I_2))$, we have
\be
c_-=\frac{a_-}{a_-+a_+},\quad
c_+=\frac{a_+}{a_-+a_+}
\ee
with
\be
a_-=\left|\pphi_\lambda(\min(J\cap\Z))\right|^2,\quad
a_+=\left|\pphi_\lambda(\max(J\cap\Z))\right|^2.
\ee
\end{corollary}
Note that $(1-\chi^\inc)\psi$ is outgoing with $\ope{supp}^\inc(1-\chi^\inc)\psi\subset\ope{chs}(C-I_2)$ for any $\psi\in\cH_{\ope{loc}}$.
\begin{figure}
\centering
\includegraphics[bb=0 0 501 248, width=10cm]{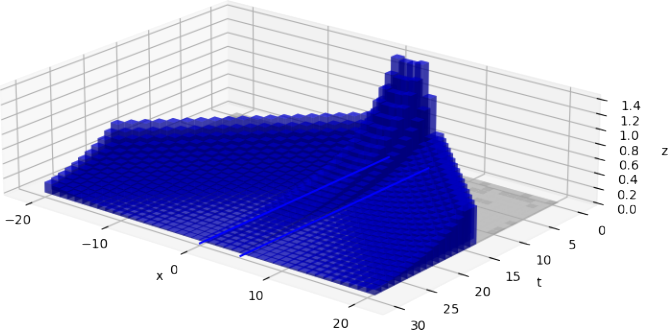}
\caption{Evolution of $\left\|U^t\mathbbm{1}_{[-1,11]}\pphi_1\right\|$ with $r=2^{-1/2}$, $k=10$}
\label{Fig:Evo_Res_St}
\end{figure}
\subsection{Background, motivation, and related works}\label{sec:Background}
Quantum walks have been studied under various motivations, such as spectral graph theory \cite{GZ}, quantum information theory \cite{Por}, and probability theory \cite{BBC}. One can also find their experimental implementations in quantum optics \cite{MW}. Recently, discrete-time quantum walks are also studied as discrete analogue of the scattering of Schr\"odinger equations \cite{FeHi,HiMo,MMOS,RST,Su}.
From the viewpoint of probability theory, quantum walks are seen as a ``quantum version" of random walks. 
As we have seen in Subsection~\ref{sec:LTB}, the resonances are deeply connected with the long-time behavior of the quantum walker. 
There are also many attempts to consider suitable models which describe diverse quantum effects. As one of such attempts, non-unitary time evolution (corresponding to a non-hermitian, non-self-adjoint Hamiltonian) is introduced for quantum walks to describe open systems \cite{MKO}. 
%
In the study of quantum mechanics, resonances and associated resonant states are one of main objects to analyze non-hermitian systems, especially the particle decay \cite{Moi}. At least for our present setting, we can observe the decaying rate in terms of resonances also for quantum walks (Corollary~\ref{cor:EstiSurvProb}). 
We also mention that there are many similarity with the problems of finite absorbing quantum walks (see e.g., \cite{BBC,Ku}). Resonances are equivalently defined as eigenvalues of a matrix of a finite size (see Corollary~\ref{cor:Evo0Res}).

Resonances have been studied in many branches of mathematics, physics, and engineering (see the survey \cite{Zw17} and references therein).
%
Among the problems in which resonances are studied, many quantities of quantum walks are concretely computable since both the space and time are discrete.
We can chase by a concrete and simple computation the dynamics of quantum walkers which is interpreted as the probability distribution of a quantum particle. We expect that descriptions of quantum effects by such a simple model will clarify what the essential property which causes the effect is. For example, Feynman and Hibbs introduced the Feynman checkerboard to explain their idea of the path integral \cite{FeyHi}. The Feynman checkerboard is considered as one of primitive models of quantum walks (see e.g., \cite{Mey} for another primitive model). Resonances for discrete models are also studied for other discrete problems such as discrete Laplacians \cite{BST,Kl}, Jacobi operators \cite{IaKo}, and discrete-time dynamical system \cite{BER} (discrete version of the Pollicott-Ruelle resonances \cite{Po,Ru}) . 

We consider the unitary time evolution operator although the generating Hamiltonian is usually considered in various situations of studying resonances. The Hamiltonian generating quantum walk is given \cite{Cha}, but at least for position dependent quantum walks, the unitary operator seems to be easier to see the properties. 
The continuation from outside the unit circle $|\lambda|>1$ to $|\lambda|\le1$ for our spectral parameter $\lambda=e^{-i\kappa}$ for the unitary time-evolution operator corresponds to that from the upper half plane $\im\kappa>0$ to $\im\kappa\le0$ for the spectral parameter (or its square root) $\kappa$ of the Hamiltonian. 
Indeed, various objects and methods in the scattering theory are rehashed in the unitary framework (see e.g., \cite{RST,ST,Ti}), especially in the study of the scattering theory on quantum walks \cite{FeHi,KKMS,MMOS,Mo,RST,Su}. We also mention that Kato and Kuroda \cite{KK} studied an abstract theory of wave operators for the discrete-time evolution given by a unitary operator in 70's.

The resonance expansion is a typical motivation to study resonances also in other settings (e.g., \cite{BuZw,NSZ,TaZw} for wave and Schr\"odinger equations, \cite{JiZw,Ts} for Anosov flows where these results are essential to show the validity of the expansion \cite[Theorem 19]{Zw17}). 
Our resulting expansion (Theorem~\ref{thm:ResExp}) is closer to that for Schr\"odinger operators in the sense that the time evolution is directly expanded, whereas the correlation is expanded in that for Anosov flows. Since the number of resonances in our setting is finite, our formula is much simpler than them. The authors will introduce an analogous method to the complex scaling in the other manuscript which allows to study resonances for the case with a perturbation not necessarily having a finite support \cite{HiMo}. 
Another reason of the simplicity is due to Formula~\eqref{eq:Commute-U1} on the time evolution of outgoing states. In general, the discontinuity of the indication function makes some ``noise". The formula is not true even for quantum walks finitely perturbed from $\til{U}_0:=SC_0$ with some position independent (unitary) non-diagonal coin $C_0$.

We also show the generic simplicity of resonances (Theorem~\ref{thm:GenSimp}). As we have seen in the previous subsections, many formulas are reduced to simpler forms when every resonance is simple. 
In the case of Laplacian, an analogue to this theorem is shown in \cite[Theorem 2.25]{DyZw} by employing the Grushin problem to reduce locally the distribution of resonances to that of zeros of a holomorphic function. In our setting, resonances are just zeros of a polynomial appearing in the transfer matrix. We can directly apply the Rouch\'e-type formula (Lemma~\ref{lem:Rouche}).

\subsection{Plan of the paper}
In the next section, we prove the main theorem on the resonance expansion (Theorem~\ref{thm:ResExp}) and preliminarily propositions to define resonances and to show the properties of the extended family of cut-off resolvent (especially, the existence of generalized resonant states).  
Section~\ref{sec:proof-LTB} is devoted for proofs of the corollaries on the long-time behavior of quantum walks stated in Subsection~\ref{sec:LTB}. 
We then see the distribution, the multiplicity, and their symmetries of resonances in  Section~\ref{sec:Dist-Res}.
We finally show the generic simplicity of resonances in Section~\ref{sec:GenSimp} by using characterizations of resonances by the transfer matrix or the scattering matrix.

\section{Proof of the resonance expansion}
In this section, we show that resonances are well-defined under Assumption and prove our main theorem (Theorem~\ref{thm:ResExp}).

\subsection{Meromorphic continuation of the resolvent 
}
We define resonances as poles of continued resolvent. We first prove the following proposition. Recall that $U$ is unitary and that each $\lambda\in\C$ with $\left|\lambda\right|>1$ belongs to the resolvent. For such $\lambda$, let $R(\lambda):=(U-\lambda)^{-1}$ be the resolvent operator on $\cH$. 
\begin{proposition}\label{prop:ContiResolv}
Under Assumption, the followings hold:
\begin{enumerate}
\item For any bounded interval $J$, the holomorphic family $\{\mathbbm{1}_JR(\lambda)\mathbbm{1}_J;\,\left|\lambda\right|>1\}$ of operators on $\cH$ can be extended meromorphically to whole $\C$. \label{item:ContiResolv}
\item For any choice of $J$ with $J\cap\Z\neq\emptyset$, $\lambda=0$ is a pole of the extended operator.  
\item The non-zero poles and their multiplicity are invariant to the choice of $J\supset\ope{chs}(C-I_2)$.\label{item:Invariance-multi}
\end{enumerate}
\end{proposition}
We formally interpret Proposition~\ref{prop:ContiResolv} (\ref{item:ContiResolv}) and (\ref{item:Invariance-multi}) as meaning that $R(\lambda)$ is continued meromorphically to $\C\setminus\{0\}$ as a family of operators from $\cH_{\ope{comp}}$ to $\cH_{\ope{loc}}$. 
In fact, the linear map $R(\lambda):\cH_{\ope{comp}}\to\cH_{\ope{loc}}$ is well-defined for $\lambda\in\C\setminus\ope{Res}(U)$. Let $\psi$ belong to $\cH_{\ope{comp}}$ and let $J_1,J_2$ be intervals containing $\ope{chs}(C-I_2)\cup\ope{supp}\,\psi$. Then we have
\be\label{eq:R-welldef}
R_{J_1}(\lambda)\psi=R_{J_2}(\lambda)\psi\qtext{on}J_1\cap J_2.
\ee
This is true since for each $x\in J_1\cap J_2$, $R_{J_j}(\lambda)\psi(x)$ is a meromorphic function with respect to $\lambda$ $(j=1,2)$. We see first that Identity \eqref{eq:R-welldef} holds for $|\lambda|>1$, and it extends by the identity theorem. 
However, the notion of analyticity is well-defined only for families of operators between Banach spaces, and ``a meromorphic family of operators $\cH_{\ope{comp}}\to\cH_{\ope{loc}}$" is not precise. 

To prove Proposition~\ref{prop:ContiResolv}, we prepare a lemma. 
Recall that for $\lambda\in\C$ with $\left|\lambda\right|>1$, the resolvent operator $R_0(\lambda)=(U_0-\lambda)^{-1}$ for the free quantum walk $U_0=S$ is given by
\be\label{eq:Free-Resolvent}
R_0(\lambda)\psi(x)=
-\sum_{y=0}^{+\infty}\lambda^{-y-1}\begin{pmatrix}\psi^L(x+y)\\\psi^R(x-y)\end{pmatrix}.
\ee
\begin{lemma}\label{lem:Free-Res-Esti}
The resolvent operator $R_0(\lambda)$ is bounded operator on $\cH$ satisfying
\be\label{eq:Free-Res-Esti1}
\left\|R_0(\lambda)\right\|_{\cH\to\cH}\le\left|\lambda\right|^{-1}(1-\left|\lambda\right|^{-2})^{-1/2}\qtext{for}\left|\lambda\right|>1.
\ee
Let $J$ be a bounded interval, then the operator $\mathbbm{1}_JR_0(\lambda)\mathbbm{1}_J$ is extended meromorphically to whole $\C$. Its unique pole is $\lambda=0$. The norm is estimated as
\be\label{eq:Free-Res-Esti2}
\left\|\mathbbm{1}_JR_0(\lambda)\mathbbm{1}_J\right\|\le
\sqrt{\left|J\right|_{\Z}}\left|\lambda\right|^{-\left|J\right|_{\Z}-1}\qquad\text{for }\ 0<\left|\lambda\right|\le1.
\ee
\end{lemma}
\begin{proof}
For any $\psi\in\cH$, we have
\ben
\left\|R_0(\lambda)\psi\right\|_{\cH}^2
\le\sum_{x\in\Z}\sum_{y=0}^{+\infty}\left|\lambda\right|^{-2y-2}\left(\left|\psi^L(x+y)\right|^2+\left|\psi^R(x-y)\right|^2\right)
=\|\psi\|_{\cH}^2\sum_{y=0}^{+\infty}\left|\lambda\right|^{-2y-2},
\een
and \eqref{eq:Free-Res-Esti1} follows. 
For $0<\left|\lambda\right|\le1$, the infinite sum in the right hand side of \eqref{eq:Free-Resolvent} does not converge in general. 
It becomes a finite sum after cut-off, and we have
\begin{align*}
\left\|\mathbbm{1}_JR_0(\lambda)\mathbbm{1}_J\psi\right\|_{\cH}^2
\le\sum_{x\in J}\sum_{y\ge 0}\left|\lambda\right|^{-2y-2}\left(\left|\mathbbm{1}_J\psi^L(x+y)\right|^2+\left|\mathbbm{1}_J\psi^R(x-y)\right|^2\right).
\end{align*}
By a change of the variable, we obtain
\begin{align*}
\sum_{x\in J}\sum_{y\ge0}\left|\lambda\right|^{-2y-2}\left|\mathbbm{1}_J\psi^L(x+y)\right|^2
&=\sum_{z\in J}\left|\psi^L(z)\right|^2\sum_{y\ge0,\,z-y\in J}\left|\lambda\right|^{-2y-2}
\\&
\le \left|J\right|_\Z\left|\lambda\right|^{-2\left|J\right|_\Z-2}\sum_{z\in\Z}\left|\psi^L(z)\right|^2.
\end{align*}
This with a symmetric estimate for the other term $\sum_{x\in J}\sum_{y\ge0}\left|\lambda\right|^{-2y-2}\left|\mathbbm{1}_J\psi^R(x-y)\right|^2$ implies \eqref{eq:Free-Res-Esti2}.
\end{proof}

The meromorphic continuation is due to the analytic Fredholm theory. We also use the Neumann series argument with the estimate given in Lemma~\ref{lem:Free-Res-Esti}.
\begin{proof}[Proof of Proposition~\ref{prop:ContiResolv}]
Note that the resolvent operator $R(\lambda)=(U-\lambda)^{-1}$ of the unitary operator $U$ on $\cH$ is well-defined for $\left|\lambda\right|>1$. 
According to the standard resolvent equation, for $\left|\lambda\right|>\sqrt{5}$, $R(\lambda)$ is expressed as 
\be\label{eq:Resolvent-eq1}
R(\lambda)=
R_0(\lambda)[I+(U-U_0)R_0(\lambda)]^{-1},
\ee
since $U-U_0$ and $R_0(\lambda)$ are bounded operators with their norm bounded by $2$ and $\left|\lambda\right|^{-1}(1-\left|\lambda\right|^{-2})^{-1/2}$, respectively (see \eqref{eq:Free-Res-Esti1}). 

Under Assumption, for any interval $J$ containing $\ope{chs}(C-I_2)$, we have 
$(I-\mathbbm{1}_J)(U-U_0)=0$ and
\ben
[I+(U-U_0)R_0(\lambda)(I-\mathbbm{1}_J)]^{-1}
=I-(U-U_0)R_0(\lambda)(I-\mathbbm{1}_J).
\een
This with the factorization 
\ben
I+(U-U_0)R_0(\lambda)=\bigl[I+(U-U_0)R_0(\lambda)(1-\mathbbm{1}_J)\bigr]
\bigl[I+(U-U_0)R_0(\lambda)\mathbbm{1}_J\bigr]
\een
implies 
\be\label{eq:Resolvent-eq2}
[I+(U-U_0)R_0(\lambda)]^{-1}
=[I+(U-U_0)R_0(\lambda)\mathbbm{1}_J]^{-1}
[I-(U-U_0)R_0(\lambda)(I-\mathbbm{1}_J)].
\ee
Combining \eqref{eq:Resolvent-eq1} and \eqref{eq:Resolvent-eq2}, we obtain the following representation of the resolvent
\be\label{eq:Resolvent-eq3}
R(\lambda)=R_0(\lambda)[I+(U-U_0)R_0(\lambda)\mathbbm{1}_J]^{-1}
[I-(U-U_0)R_0(\lambda)(I-\mathbbm{1}_J)].
\ee

We see that for any intervals $J_1$ and $J_2\supset \mathcal{N}_1(J\cup J_1)$, 
\be\label{eq:Compactly-supported}
\begin{aligned}
&(I-\mathbbm{1}_{J\cup J_1})[I-(U-U_0)R_0(\lambda)(I-\mathbbm{1}_J)]\mathbbm{1}_{J_1}=0,\\
&(I-\mathbbm{1}_{J_2})[I+(U-U_0)R_0(\lambda)\mathbbm{1}_J]^{-1}\mathbbm{1}_{J\cup J_1}=0,
\end{aligned}\qquad\text{for }\ \left|\lambda\right|>\sqrt{5}.
\ee
The second one follows from the Neumann series representation
\be\label{eq:Neumann-rep}
[I+(U-U_0)R_0(\lambda)\mathbbm{1}_J]^{-1}
=\sum_{k\ge0}(-(U-U_0)R_0(\lambda)\mathbbm{1}_J)^k.
\ee
According to Lemma~\ref{lem:Free-Res-Esti}, the cut-off free resolvent $\mathbbm{1}_JR_0(\lambda)\mathbbm{1}_J$ is a meromorphic family of operators for $\lambda\in\C$. 
The operator $(U-U_0)R_0(\lambda)\mathbbm{1}_J=(U-U_0)\mathbbm{1}_JR_0(\lambda)\mathbbm{1}_J$ is of finite rank, in particular compact, and hence the analytic Fredholm theory (see e.g., \cite[Theorem~C.8 and C.10]{DyZw}) shows that $[I+(U-U_0)R_0(\lambda)\mathbbm{1}_J]^{-1}$ is extended meromorphically to $\lambda\in\C$. 
By the identity theorem, the properties of the support \eqref{eq:Compactly-supported} still hold after the extension. Finally, the cut-off resolvent
\be
\mathbbm{1}_{J_1}R(\lambda)\mathbbm{1}_{J_1}=\mathbbm{1}_{J_1}R_0(\lambda)\mathbbm{1}_{J_2}[I+(U-U_0)R_0(\lambda)\mathbbm{1}_J]^{-1}\mathbbm{1}_{J\cup J_1}
[I-(U-U_0)R_0(\lambda)(I-\mathbbm{1}_J)]\mathbbm{1}_{J_1}
\ee
is extended meromorphically to $\lambda\in\C$.
In addition, the poles other than $\lambda=0$ of this operator come from $[I+(U-U_0)R_0(\lambda)\mathbbm{1}_J]^{-1}$. 
The Neumann series representation \eqref{eq:Neumann-rep} shows that for any two intervals $J_3$ and $J_4$ containing $J$, we have
\ben
\left[I+(U-U_0)R_0(\lambda)\mathbbm{1}_J\right]^{-1}(\mathbbm{1}_{J_3}-\mathbbm{1}_{J_4})
=\mathbbm{1}_{J_3}-\mathbbm{1}_{J_4}\qquad\text{for }\ \left|\lambda\right|>\sqrt{5}.
\een
This invariance is true for all $\lambda\in\C$. 
Therefore, the only dependence on the choice of interval larger than $J$ of the poles is the order of the pole at $\lambda=0$.
\end{proof}

\subsection{Properties of the outgoing resolvent and resonant states}

We observe some properties of the continued resolvent. It is characterized as the operator assigning the unique outgoing solution for the equation 
\be\label{eq:EigenEquation}
(U-\lambda)\psi=f
\ee
to each $f\in\cH_{\ope{comp}}$ (Proposition~\ref{thm:DefRes} (\ref{enu:Reg})). 
Under Assumption, the above equation is trivial for $\left|x\right|\gg1$, and for each solution $\psi$, there exist four constants $c_\pm^\out$ and $c_\pm^\inc$ such that 
\be\label{eq:Outgoing-Cond}
\psi(x)=\left\{
\begin{aligned}
&\begin{pmatrix}c_-^\out\lambda^x\\
c_-^\inc\lambda^{-x}
\end{pmatrix}
=c_-^\out \Psi_L(\lambda,x)+c_-^\inc \Psi_R(\lambda,x)
&&\text{for }-x\gg1,\\
&
\begin{pmatrix}
c_+^\inc\lambda^x\\
c_+^\out\lambda^{-x}
\end{pmatrix}
=c_+^\out \Psi_R(\lambda,x)+c_+^\inc \Psi_L(\lambda,x)
&&\text{for }+x\gg1.
\end{aligned}\right.
\ee
Here, $\Psi_L(\lambda,\cdot)$ and $\Psi_R(\lambda,\cdot)$ are left and right going solution to the non perturbed equation $(U_0-\lambda)f(\lambda,\cdot)=0$ defined by
\ben
\Psi_L(\lambda,x)=\begin{pmatrix}\lambda^x\\0\end{pmatrix},\quad
\Psi_R(\lambda,x)=\begin{pmatrix}0\\\lambda^{-x}\end{pmatrix}.
\een
More precisely, $\pm x\gg1$ means outside the convex full of $\ope{supp}\,f\cup\ope{supp}(C-I_2)$. 
In this case, $\psi$ is outgoing if and only if  $c_+^\inc=c_-^\inc=0$. 
In particular, for $\left|\lambda\right|>1$, $\mathbbm{1}_{(-\infty,0]}\Psi_L(\lambda,\cdot)$ and $\mathbbm{1}_{[0,+\infty)}\Psi_R(\lambda,\cdot)$ belong to $\cH$. Then the outgoing condition is also equivalent for $\psi$ to belonging to $\cH$.

\begin{proposition}\label{thm:DefRes}
Under Assumption, the followings are true:
\begin{enumerate}
\item \label{enu:Reg}For any $\lambda\in\C\setminus\ope{Res}(U)$ and for any $\pphi\in\cH_{\ope{comp}}$, $\psi:=R(\lambda)\pphi$ is the unique outgoing solution to \eqref{eq:EigenEquation}. 
In particular, $ R(\lambda)\pphi\notin\cH$ for $\left|\lambda\right|<1$. The incoming support $\ope{supp}^\inc (R(\lambda)\pphi)$ is a subset of the convex hull of the union 
\be\label{eq:Contain-inc-supp}
\mathcal{N}_1(\ope{supp}\,\pphi)\cup\ope{supp}(C-I_2).
\ee 
\item \label{enu:Res-State}A complex number $\lambda\in\C\setminus\{0\}$ is a resonance if and only if there exists a non-identically vanishing, (unbounded) outgoing solution $\pphi_\lambda\in\cH_{\ope{loc}}$ to
\be\label{eq:Eigen-Equation}
(U-\lambda)\pphi_\lambda=0.
\ee
\item \label{enu:Jor-Chain} There exists a tuple $(\pphi_{\lambda,k})_{k=1,2,\ldots,m(\lambda)}$ of outgoing maps for each non-zero resonance $\lambda$ such that
\be
(U-\lambda)\pphi_{\lambda,k}=\pphi_{\lambda,k-1},\quad\mathcal{N}_1(\ope{supp}^\inc\pphi_{\lambda,k})\subset\ope{chs}(C-I_2)\quad(1\le k\le m(\lambda)),\quad 
\ee
holds with $\pphi_{\lambda,0}=0$. In particular, $\pphi_{\lambda,1}$ is unique up to a multiplicative constant (i.e., the geometric multiplicity of each non-zero resonance is one).
\item \label{enu:Sum-Multi} The number of non-zero resonances is bounded by $2(\left|\ope{chs}(C-I_2)\right|_\Z-1)$, where we count each resonance $\lambda$ the same time as its multiplicity $m(\lambda)$. 
\end{enumerate}
\end{proposition}


\begin{proof}
(\ref{enu:Reg}) Let $f\in\cH_{\ope{comp}}$ and $\left|\lambda\right|>1$. Note that for $\left|\lambda\right|>1$ and for a solution $\psi$ to \eqref{eq:EigenEquation}, the outgoing condition is equivalent to that $\psi$ belongs to $\cH$. 
By definition, we have 
\ben
(U-\lambda)R(\lambda)f=f,
\quad
(R(\lambda)f)^L(x_+)=(R(\lambda)f)^R(x_-)=0
\een 
for $x_\pm\in\Z$ such that $(x_-,x_+)\supset \mathcal{N}_1(\ope{supp}\,f)\cup\ope{chs}(C-I_2)$. 
These identities extend to all $\lambda$ by analyticity. 
For the uniqueness, let $\lambda_0\in\C\setminus\ope{Res}(U)$. It suffices to show the identity
\be\label{eq:Suff-Unique}
\psi=R(\lambda_0)(U-\lambda_0)\psi
\ee
for any $\psi\in\cH_{\ope{loc}}$ such that there exist constants $c_\pm$ such that
\be\label{eq:lambda0-outgoing}
\psi(x)=\left\{
\begin{aligned}
&c_- \Psi_L(\lambda_0,x)
&&\text{for }x\ll-1,\\
&c_+ \Psi_R(\lambda_0,x)
&&\text{for }x\gg1.
\end{aligned}\right.
\ee
See \eqref{eq:Outgoing-Cond} for the definition of $\Psi_L$ and $\Psi_R$. In fact, let $\psi_1,\psi_2$ be two outgoing solutions to \eqref{eq:EigenEquation} for $f\in\cH_{\ope{comp}}$. 
Then it follows that $R(\lambda_0)(U-\lambda_0)\psi_1=R(\lambda_0)(U-\lambda_0)\psi_2=R(\lambda_0)f$. This with \eqref{eq:Suff-Unique}  implies the uniqueness: $\psi_1=\psi_2$. 

To show the identity \eqref{eq:Suff-Unique}, we decompose given $\psi$ of the form \eqref{eq:lambda0-outgoing} into three parts,
\ben
\psi=\mathbbm{1}_{(-\infty,-r)}\psi+\mathbbm{1}_{[-r,r]}\psi+\mathbbm{1}_{(r,+\infty)}\psi,
\een
where $r\gg1$ is so taken that $(-r,r)\supset\ope{chs}(C-I_2)$ and 
\ben
\mathbbm{1}_{(-\infty,-r)}\psi(x)=c_-\mathbbm{1}_{(-\infty,-r)}\Psi_L(\lambda_0,x),\quad
\mathbbm{1}_{(r,+\infty)}\psi(x)=c_+\mathbbm{1}_{(r,+\infty)}\Psi_R(\lambda_0,x),
\een
holds with some constants $c_\pm$ for any $x\in\Z$. 
For $\lambda\in\C$ with $\left|\lambda\right|>1$, we have
\begin{align*}
&R(\lambda)(U-\lambda)\mathbbm{1}_{[-r,r]}\psi=\mathbbm{1}_{[-r,r]}\psi,\quad
R(\lambda)(U-\lambda)\mathbbm{1}_{(-\infty,-r)}\Psi_L(\lambda,\,\cdot\,)=\mathbbm{1}_{(-\infty,-r)}\Psi_L(\lambda,\,\cdot\,),\\
&R(\lambda)(U-\lambda)\mathbbm{1}_{(r,+\infty)}\Psi_R(\lambda,\,\cdot\,)=\mathbbm{1}_{(r,+\infty)}\Psi_R(\lambda,\,\cdot\,).
\end{align*}
Since $(U-\lambda)\mathbbm{1}_{[-r,r]}\psi$, $(U-\lambda)\mathbbm{1}_{(-\infty,-r)}\Psi_L(\lambda,\,\cdot\,)$, and $(U-\lambda)\mathbbm{1}_{(r,+\infty)}\Psi_R(\lambda,\,\cdot\,)$ are compactly supported, the above identities are valid for $\lambda$ at which $R(\lambda)$ is holomorphic, in particular at $\lambda=\lambda_0$. Then \eqref{eq:Suff-Unique} follows from this with the linearity.

\noindent
(\ref{enu:Res-State}) 
Let $\lambda_0\in\ope{Res}(U)\setminus\{0\}$ and $\mu:=\sqrt{\lambda}$. By the meromorphic continuation, there exist $K\ge1$, finite rank operators $A_1,\ldots,A_K$, and a holomorphic family $A_0(\lambda)$ such that
\ben
R(\mu^2)=\sum_{k=1}^K \frac{A_k}{(\mu^2-\lambda_0)^k}+A_0(\mu^2),
\een
holds for $\mu$ near $\sqrt{\lambda_0}$. 
By the residue theorem, we have
\ben
A_1=-\Pi_{\lambda_0}:=\frac1{2\pi i}\oint_{\lambda_0}R(\lambda)d\lambda=\frac1{2\pi i}\oint_{\sqrt{\lambda_0}}R(\mu^2)2\mu d\mu,
\een
and
\be\label{eq:Residue-Square}
\frac1{2\pi i}\oint_{\sqrt{\lambda_0}}R(\mu^2)d\mu=
\sum_{k=1}^K(-1)^{k-1}\frac{(2k-2)!}{(k-1)!}(2\lambda_0)^{-2k+1}A_k.
\ee
Since $(U-\mu^2)R(\mu^2)=\ope{Id}_{\cH_{\ope{comp}}}$ near $\mu=\sqrt{\lambda_0}$, modulo terms holomorphic near $\sqrt{\lambda_0}$ we have
\ben
0\equiv(U-\mu^2)R(\mu^2)
\equiv\sum_{k=1}^K\frac{(U-\lambda_0)A_k-A_{k+1}}{(\mu^2-\lambda_0)^k},
\een
where we use the convention that $A_{K+1}=0$. 
It follows that $A_{k+1}=(U-\lambda_0)A_k$ for $k=1,\ldots,K$. 
As a consequence, we obtain
\ben
(U-\lambda_0)^K\Pi_{\lambda_0}=-(U-\lambda_0)^KA_1=0.
\een
Since the operator $U-\lambda_0$ commutes with $\Pi_{\lambda_0}$, $U-\lambda_0$ is nilpotent on $\ope{Ran}\Pi_{\lambda_0}$. 
Hence we can express it by a Jordan normal form. 
There exists a basis $\{\pphi_{l,j};1\le l\le L,\,1\le j\le k_l\}$ of $\ope{Ran}\Pi_{\lambda_0}$ such that $\sum_{l=1}^Lk_l=K$ and 
\begin{align*}
(U-\lambda_0)\pphi_{l,j}=\pphi_{l,j-1},\quad(1\le j\le k_l,\ \pphi_{l,0}=0)
\end{align*}
holds for each $l$. 
Since $\pphi_{l,j}=R_0(\lambda_0)(\pphi_{l,j-1}-(U-U_0)\pphi_{l,j})$, each $\pphi_{l,j}$ belongs to 
\ben
\sum_{k=1}^j R_0(\lambda_0)^k(\cH_{\ope{comp}})
=\left\{\sum_{k=1}^{j} R_0(\lambda_0)^k\psi_k;\,\psi_k\in\cH_{\ope{comp}}\right\}.
\een 
In particular, $\pphi_{l,1}\in R_0(\lambda)(\cH_{\ope{comp}})$ implies that $\pphi_{l,1}$ is an outgoing solution to an equation of type \eqref{eq:EigenEquation} with $U=U_0$. This with the uniqueness of the continuation of solutions to $(U-\lambda)\psi=0$ implies that $\pphi_{l,1}$ is unique up to a multiplicative constant. Therefore, there is only one Jordan chain, i.e., $L=1$, and $k_1=K=\dim\ope{Ran}\Pi_{\lambda_0}$. This also proves (\ref{enu:Jor-Chain}) with 
\ben
m(\lambda_0)=\ope{rank}\oint_{\lambda_0}R(\lambda)d\lambda=\ope{rank}\Pi_{\lambda_0}=K.
\een
The inclusion of the incoming support follows from $U\pphi_{l,j}=\lambda_0\pphi_{l,j}+\pphi_{l,j-1}$ by an induction with respect to $j$. Suppose that there existed $x_0\le \ope{chs}(C-I_2)$ such that $\pphi_{l,j}^R(x_0)\neq0$. Then it would follow from the above identity and $(U\pphi_{l,j})^R(x_0)=\pphi_{l,j}^R(x_0-1)$ that 
\ben
\pphi_{l,j}^R(x_0-1)=
\lambda_0\pphi_{l,j}^R(x_0)+\pphi_{l,j-1}^R(x_0).
\een
The induction hypothesis $\pphi_{l,j-1}(x_0)=0$ would imply that $\pphi_{l,j}^R(x_0-1)\neq0$, and $\pphi_{l,j}$ is not outgoing. This is a contradiction.

Conversely, suppose that there exists a resonant state $\pphi_{\lambda_0}$ for $\lambda_0\in\C\setminus\{0\}$.  The uniqueness \eqref{eq:Suff-Unique} shows 
\ben
\pphi_{\lambda_0}=R(\lambda)(U-\lambda)\pphi_{\lambda_0}
=(\lambda_0-\lambda)R(\lambda)\pphi_{\lambda_0},
\een 
for $\lambda\in\C\setminus\ope{Res}(U)$, and 
\ben
\oint_{\lambda_0}R(\lambda)\pphi_{\lambda_0}d\lambda
=\oint_{\lambda_0}\frac{d\lambda}{\lambda_0-\lambda}\pphi_{\lambda_0}
=-2\pi i \pphi_{\lambda_0}.
\een

\noindent
(\ref{enu:Sum-Multi}) 
Since the resolvent equation $(\lambda'-\lambda)R(\lambda)R(\lambda')=R(\lambda')-R(\lambda)$ extends analytically to $\lambda\in\C\setminus\ope{Res}(U)$, for any resonances $\lambda_1,\lambda_2\in\ope{Res}(U)$ with $\lambda_1\neq\lambda_2$, the projections $-(2\pi i)^{-1}\oint_{\lambda_j} R_J(\lambda')d\lambda'$ $(j=1,2)$ are mutually orthogonal.
Thus, by Cauchy's integral theorem, we have
\ben
\ope{rank}\oint_{\left|\lambda\right|=1}R_J(\lambda)d\lambda
=\sum_{\lambda\in\ope{Res}(U)}\ope{rank}\oint_\lambda R_J(\lambda')d\lambda'
=\ope{rank}\oint_0 R_J(\lambda)d\lambda+\sum_{\lambda\in\ope{Res}(U)\setminus\{0\}}m(\lambda).
\een
Since $R_J(\lambda)$ is analytic for $|\lambda|>1$, we also have $\ope{rank}\oint_{\left|\lambda\right|=1}R_J(\lambda)\,d\lambda=\ope{rank}\mathbbm{1}_J=2\left|J\right|_\Z$. Recall that $m(\lambda)$ is independent of the choice of $J\supset\ope{chs}(C-I_2)$. Let us take $J=\ope{chs}(C-I_2)$. We will see in Remark~\ref{rem:dim-Zero-GES} that $\ope{rank}\oint_0R_J(\lambda)d\lambda\ge2$. This ends the proof.
\end{proof}

\subsection{Time evolution of resonant states}
The time evolution of restricted (generalized) resonant states is given by the following proposition.
\begin{proposition}\label{prop:EvoResSt}
Let $\lambda\in\C\setminus\{0\}$ be a resonance. 
For any interval $J\supset\ope{chs}(C-I_2)$ and $n\in\N$, we have
\be
U^n(\mathbbm{1}_{J}\pphi_{\lambda})=\lambda^n(\mathbbm{1}_{\mathcal{N}_n(J)}\pphi_{\lambda}),
\ee
for an associated resonant state $\pphi_\lambda$, and
\be
U^n(\mathbbm{1}_{J}\pphi_{\lambda,k})=\lambda^n\left(\mathbbm{1}_{\mathcal{N}_n(J)}\sum_{l=0}^{k-1}\begin{pmatrix}n\\l\end{pmatrix}\lambda^{-l}\pphi_{\lambda,k-l}\right),
\ee
for a generalized resonant state $\pphi_{\lambda,k}$, where $(\pphi_{\lambda,k})_{k=1,2,\ldots,m(\lambda)}$ forms a Jordan chain. we here use a convention $\begin{pmatrix}n\\l\end{pmatrix}=0$ for $n<l$.
\end{proposition}

Proposition~\ref{prop:EvoResSt} is a straightforward consequence of the following lemma with the fact that each generalized resonant state is outgoing with its incoming support contained in $\ope{chs}(C-I_2)$ (Proposition~\ref{thm:DefRes} (\ref{enu:Jor-Chain})).

\begin{lemma}\label{lem:Outgoing-TimeEv}
For any outgoing map $\psi\in\cH_{\ope{loc}}$ and for any interval $J\supset\ope{supp}^\inc\psi\cup\ope{chs}(C-I_2)$, we have
\be\label{eq:Outgoing-TimeEv}
U\mathbbm{1}_{J}\psi=\mathbbm{1}_{\mathcal{N}_1(J)}U\psi.
\ee
\end{lemma}
\begin{proof}
By definition, for each $x\in\Z$, $(U\psi)(x)$ depends only on $\psi(x-1)$ and $\psi(x+1)$. This implies \eqref{eq:Outgoing-TimeEv} away from the extremal points of $J=:[a,b]$, namely for $x\in\Z\setminus\{a-1,a, b, b+1\}$. Recall that the dependence on $\psi(x-1)$ and $\psi(x+1)$ is shown explicitly by
\ben
(U\psi)(x)=\begin{pmatrix}1&0\\0&0\end{pmatrix}C(x+1)\psi(x+1)+\begin{pmatrix}0&0\\0&1\end{pmatrix}C(x-1)\psi(x-1).
\een
Note that $a-1\notin J$ implies that $C(a-2)=C(a-1)=I_2$ and $\psi^R(a-2)=\psi^R(a-1)=0$. As a consequence, for $x\in\{a-1,a\}$, we have
\ben
(U\psi)(x)
=\begin{pmatrix}1&0\\0&0\end{pmatrix}C(x+1)\psi(x+1)=(U\mathbbm{1}_J\psi)(x).
\een
We also obtain $(U\psi)(x)=(U\mathbbm{1}_J\psi)(x)$ for $x\in\{b,b+1\}$ in the same way. 
In particular, this also holds for $x=a-1$ and $x=b+1$ even though they do not belong to $J$ but $\mathcal{N}_1(J)$.
\end{proof}


\subsection{Resonance expansion}
Following lemma shows that the non-zero resonances are eigenvalues of a finite rank operator.
\begin{lemma}\label{lem:cut-off-Res}
For any interval $J\supset\ope{chs}(C-I_2)$, we have $R_J(\lambda)=(\mathbbm{1}_JU\mathbbm{1}_J-\lambda)^{-1}$. 
\end{lemma}
\begin{proof}
For $\lambda\notin\ope{Res}(U)$, one has
\be
\mathbbm{1}_J=\mathbbm{1}_J(U-\lambda)R(\lambda)\mathbbm{1}_J=\mathbbm{1}_J(U-\lambda)\mathbbm{1}_JR(\lambda)\mathbbm{1}_J
+\mathbbm{1}_J[\mathbbm{1}_J,U]R(\lambda)\mathbbm{1}_J.
\ee
It suffices to show that the second term is zero for $\lambda\in\C\setminus\ope{Res}(U)$:
\be\label{eq:Remainder-Commutator}
\mathbbm{1}_J[\mathbbm{1}_J,U]R(\lambda)\mathbbm{1}_J=0.
\ee
According to \eqref{eq:Contain-inc-supp},  $R(\lambda)\mathbbm{1}_J\psi$ is outgoing with $\ope{supp}^\inc (R(\lambda)\mathbbm{1}_J\psi)\subset\mathcal{N}_1(J)$ for any $\psi\in\cH_{\ope{loc}}$. Then Lemma~\ref{lem:Outgoing-TimeEv} implies that
\ben
[\mathbbm{1}_J,U]R(\lambda)\mathbbm{1}_J\psi=(\mathbbm{1}_{\mathcal{N}_1(J)}-\mathbbm{1}_J) UR(\lambda)\mathbbm{1}_J\psi.
\een
Clearly, this is supported only on $\mathcal{N}_1(J)\setminus J$, and \eqref{eq:Remainder-Commutator} follows.
\end{proof}

\begin{corollary}\label{cor:Evo0Res}
A complex number $\lambda\in\C$ is a resonance of $U$ if and only if it is an eigenvalue of $\mathbbm{1}_JU\mathbbm{1}_J$. 
For $\lambda\in\ope{Res}(U)$, the generalized eigenspace associated with the eigenvalue $\lambda$ of $\mathbbm{1}_JU\mathbbm{1}_J$ is given by the range 
\ben
V_J(\lambda):=\ope{Ran}\oint_\lambda R_J(\lambda')d\lambda'
=\left\{\mathbbm{1}_J\pphi;\,\pphi\text{ is a generalized resonant state associated with }\lambda\right\}. 
\een
In particular, a state $\psi\in\cH$ with $\ope{supp}\,\psi\subset J$ belongs to $V_J(0)$ if and only if 
\be\label{eq:Zero-GES}
U^n\psi(x)=0\qtext{on}J,
\ee
for 
any $n>2\left|J\right|_\Z$. 
\end{corollary}

\begin{proof}[Proof of Theorem~\ref{thm:ResExp}]
According to Corollary~\ref{cor:Evo0Res}, the resonance expansion \eqref{eq:ResExp-Sum} is nothing but the expansion by (generalized) eigenvectors. 
The time evolution \eqref{eq:ResExp-Time} is just a summation of that of each generalized resonant states given by Proposition~\ref{prop:EvoResSt}.
\end{proof}

\begin{remark}\label{rem:dim-Zero-GES}
We have $\dim V_J(0)\ge \dim V_{\ope{chs}(C-I_2)}\ge2$ for any $J\supset\ope{chs}(C-I_2)$. It is a consequence of the fact that 
\ben
\det C^+
=\det C^-
=0,
\een
with
\ben
C^+:=\begin{pmatrix}1&0\\0&0\end{pmatrix}C(\max\ope{chs}(C-I_2)),\quad
C^-:=\begin{pmatrix}0&0\\0&1\end{pmatrix}C(\min\ope{chs}(C-I_2)).
\een
Let $v^\pm$ be an eigenvector associated with the zero eigenvalue of $C^\pm$. Then 
\ben
\pphi:=c^+\mathbbm{1}_{\max\ope{chs}(C-I_2)}v^++c^-\mathbbm{1}_{\min\ope{chs}(C-I_2)}v^-
\een
satisfies $\mathbbm{1}_JU\mathbbm{1}_J\pphi=0$ for any constants $c^\pm$.
\end{remark}
\begin{remark}
The bound $n>2\left|J\right|_\Z$ for \eqref{eq:Zero-GES} is optimal. For example, let $C(0)$ be the only non-diagonal coin, and let $J=[0,N]$ for some $N\ge1$. Then for the state $\psi(x)={}^t(\dl_{N,x},0)$, we have
\ben
\mathbbm{1}_JU^n\psi\not\equiv0\quad(n\le2N),\quad
\mathbbm{1}_JU^{n}\psi\equiv0\quad(n>2N).
\een
Remark that $V_J(0)=\{\sum_{n=0}^{2N} c_nU^n\psi+c\psi^-;\, c_n,c\in\C\}$ and $\dim V_J(0)=2(\left|J\right|_\Z+1)$ (see also Proposition~\ref{prop:no-res} for the non-existence of non-zero resonancs).
\end{remark}



\section{Proof of long-time behavior}\label{sec:proof-LTB}
We prove the corollaries stated in Subsection~\ref{sec:LTB}. They follows almost trivially from the resonance expansion (Theorem~\ref{thm:ResExp}).

By applying the triangular inequality 
to the resonance expansion \eqref{eq:ResExp-Time}, we obtain the decaying rate of Corollary~\ref{cor:EstiSurvProb}. Note that for $\lambda_0\in\ope{Res}(U)\setminus\{0\}$, the operator 
\ben
-\frac1{2\pi i}\oint_{\lambda_0}R(\lambda)d\lambda
\een
is the projection to the space of generalized resonant states associated with $\lambda_0$ (see proof of Proposition~\ref{thm:DefRes} (\ref{enu:Res-State})). 

Corollary~\ref{cor:PtWiseAsym} is also a straightforward consequence of the resonance expansion. Since each resonant state is outgoing solution to the eigenequation \eqref{eq:Eigen-Equation} (see definition and Proposition~\ref{thm:DefRes} (\ref{enu:Res-State})), its behavior outside $\ope{chs}(C-I_2)$ is given by \eqref{eq:Outgoing-Cond} with $c^\inc_\pm=0$.

Corollary~\ref{cor:Mean-time} follows from the estimates~\eqref{eq:DR1} and \eqref{eq:DR2}, where $\|\mathbbm{1}_JU^n\psi\|_{\cH}^2$ gives the probability that survival time is longer than $n+1$. Hence, we have
\ben
\tau=\sum_{n=1}^\infty n\mu_n(\mathcal{N}_1(J)\setminus J)
=\sum_{n=1}^\infty n\left(\|\mathbbm{1}_JU^{n-1}\psi\|_{\cH}^2-\|\mathbbm{1}_JU^n\psi\|_{\cH}^2\right)\le
\sum_{n=1}^\infty n\|\mathbbm{1}_JU^{n-1}\psi\|_{\cH}^2.
\een
We also used the following inequality for $l\in\N$, $n\ge l+1$, and $r>0$,
\ben
n\begin{pmatrix}n-1\\l\end{pmatrix}^2r^{2(n-1)}\le 2^{-2l-1}r^{2l-1}\frac{d^{2l+1}}{dr^{2l+1}}r^{2n}=2^{-2l-1} \Upsilon_{2l+1}(r).
\een

\begin{proof}[Proof of Corollary~\ref{cor:w-Lim}]
Since the initial state $\psi$ has a compact support, so does $U^n\psi$ with $\ope{supp}(U^n\psi)\subset\mathcal{N}_n(\ope{supp}\psi)$. This implies that for any $\epsilon>0$, there exists $n_{\epsilon}\in\N$ such that
\ben
U^n\psi(x)=0
\een
holds for any $n\ge n_{\epsilon}$ and any $x\in\Z\setminus[-(1+\epsilon)n,(1+\epsilon)n]$. Hence the density function $w$ satisfies
\be\label{eq:w-infty}
\ope{supp}\,w\subset[-1,1].
\ee
Conversely, inside $[-n,n]$, we have the resonance expansion. 
It also shows for any initial state $\psi\in\cH_{\ope{comp}}$ that there exists $n(\psi)\in\N$ such that $U^n\psi$ is outgoing with $\mathcal{N}_1(\ope{supp}^\inc(U^n\psi))\subset\ope{chs}(C-I_2)$ for any $n\ge n(\psi)$. This with Lemma~\ref{lem:Outgoing-TimeEv} implies that for any interval $J\supset\ope{chs}(C-I_2)$ and $n\ge n(\psi)$, we have
\be\label{eq:Comp-Norm}
\|\mathbbm{1}_{\mathcal{N}_{n-n(\psi)(J)}}U^n\psi\|_\cH=\|U^{n-n(\psi)}(\mathbbm{1}_JU^{n(\psi)}\psi)\|_\cH=\|\mathbbm{1}_JU^{n(\psi)}\psi\|_\cH.
\ee
On the other hand, the time evolution outside $\ope{chs}(C-I_2)$ is trivial:
\ben
U^{n+k}\psi(x)=U^n\psi(x\mp k)
\quad\text{for }k\ge0,\ \pm x\gg1,\ (x\mp k)\notin\ope{chs}(C-I_2).
\een
Then for any $\epsilon>0$, the values $\left\|\mathbbm{1}_{[(1-\epsilon)n,(1+\epsilon)n]}U^n\psi\right\|_{\cH}$ and $\left\|\mathbbm{1}_{[-(1+\epsilon)n,-(1-\epsilon)n]}U^n\psi\right\|_{\cH}$ monotonically increasing with respect to $n\ge n(\psi)$. In fact, we have
\begin{align*}
\left\|\mathbbm{1}_{[(1-\epsilon)(n+k),(1+\epsilon)(n+k)]}U^{n+k}\psi\right\|_{\cH}
&=\left\|\mathbbm{1}_{[(1-\epsilon)(n+k),(1+\epsilon)(n+k)]}U^n\psi(\cdot-k)\right\|_{\cH}
\\&\ge\left\|\mathbbm{1}_{[(1-\epsilon)n,(1+\epsilon)n]}U^n\psi\right\|_{\cH},
\end{align*}
since $[(1-\epsilon)n+k,(1+\epsilon)n+k]\subset[(1-\epsilon)(n+k),(1+\epsilon)(n+k)]$.

By combining \eqref{eq:Comp-Norm} and the monotonicity, we obtain 
\ben
\left\|\mathbbm{1}_{[-(1-\epsilon)n,(1-\epsilon)n]}U^n\psi\right\|\to0,
\een
as $n\to+\infty$. This with \eqref{eq:w-infty} implies 
\ben
\ope{supp}\,w\subset\{\pm1\},
\een
and in particular, the existence of $c_\pm$ in Corollary~\ref{cor:w-Lim}. 

The equality 
\ben
\chi^\inc U^n\psi=\mathbbm{1}_{\ope{chs}(C-I_2)}U^n\psi
\een
holds for $n\ge n(\psi)$. Thus the estimate of \eqref{eq:Conv-Spd} follows from Corollary~\ref{cor:EstiSurvProb}.  
\end{proof}

\section{Distribution of resonances}\label{sec:Dist-Res}
In the previous sections, we have seen the usefulness of resonances for quantum walks. 
In this section, we study where the resonances distribute and when they have multiplicity.
\subsection{Concrete examples}
We start with simple models such that all the resonances are computable. 
\begin{proposition}\label{prop:no-res}
If the number of non-diagonal coins $\ope{Card}\{x\in\Z;\,C(x)\text{ is not diagonal}\}$ is less than two, then there is no resonance other than zero: $\ope{Res}(U)=\{0\}$.
\end{proposition}
\begin{proof}
In this setting, there is no outgoing solution (except the zero function) to \eqref{eq:Eigen-Equation} for any $\lambda\in\C$.  This fact can be shown by using the method of transfer matrices (in particular, the transfer matrix $T_\lambda(x)$ introduced in \eqref{eq:T-kappa} is diagonal if $C(x)$ is diagonal).
Then the statement follows from Proposition~\ref{thm:DefRes} (\ref{enu:Res-State}). 
\end{proof}

According to Proposition~\ref{prop:no-res}, the double barrier problem is the simplest case with non-zero resonances. The distribution of resonances reflects a quasi-periodic dynamics of quantum walkers.
\begin{proposition}\label{prop:DoubleBarrier}
Assume that $C(x)$ is a diagonal matrix for  $x\in\Z\setminus\{0,N\}$ for some $N\in\N\setminus\{0\}$. 
Then there exists a constant $\alpha\in\C$ such that 
for any initial state $\psi\in\cH$ with $\ope{supp}^\inc\psi\subset[1,N-1]$, one has
\be
\psi_{n+2N}(x)=\alpha\psi_n(x)\qtext{for any}n\in\N,\ x\in[0,N]\cap\Z.
\ee
The set of resonances $\ope{Res}(U)$ consists of $0$ and the $2N$-roots of $\lambda^{2N}=\alpha$, which are simple. 
\end{proposition}

Without loss of generalities, we can assume that $\ope{chs}(C-I_2)=[0,N]$.

\begin{remark}\label{rem:quasi-periodic}
The constant $\alpha$ is the ``probability amplitude" associated with the dynamics. Let $\psi(1)={}^t(0,1)$ and $\psi(x)=0$ for $x\in\Z\setminus\{1\}$. Then for $0\le n\le N-1$, we have $U^n\psi(n+1)={}^t(0,a_n)$ and $\psi(x)=0$ for $x\in\Z\setminus\{n+1\}$ with 
\be\label{eq:Def-an}
a_{n+1}=a_n c_{22}(n)
\quad a_0=1,
\ee
where $c_{jk}(x)$ stands for the $(j,k)$-entry of $C(x)$. 
After that, for $N\le n\le 2N-1$, we have $U^n\psi(N-n-1)={}^t(b_n,0)$, $U^n\psi(x)=0$ for $x\in\Z\setminus\{N-n-1,n+1\}$ with 
\be\label{eq:Def-bn}
b_{n+1}=b_n
c_{11}(N-n),\quad b_N=a_{N-1}
c_{12}(N).
\ee
Then at the time $2N$, we have
\ben
U^{2N}\psi(1)=\begin{pmatrix}0\\\alpha\end{pmatrix},\quad U^{2N}\psi(x)=0\quad x\in\Z\setminus\{\pm1\}\qtext{with }
\alpha=b_{2N-1}c_{21}(0)
\een
Therefore, we have
\be\label{eq:Def-Alpha}
\alpha=c_{21}(0)c_{12}(N)
\prod_{x=1}^{N-1}\det C(x).
\ee
Let $\lambda$ be one of the roots of $\lambda^{2N}=\alpha$, where $\alpha$ is given by \eqref{eq:Def-Alpha}. 
Put
\ben
\pphi_\lambda(x):=\begin{pmatrix}\mathbbm{1}_{(-\infty,N-1]}(x)\lambda^{x-2N}b_{2N-1-x}\\\mathbbm{1}_{[1,+\infty)}(x)\lambda^{-x}a_{x}\end{pmatrix},
\een
where $a_n$ and $b_n$ are defined by \eqref{eq:Def-an} and by \eqref{eq:Def-bn}, respectively. We can easily see that $\pphi_\lambda$ is a resonant state associated with $\lambda$. 
\end{remark}

\begin{proof}
According to Lemma~\ref{lem:cut-off-Res}, the non-zero resonances of $U$ are the non-zero eigenvalues of a $2N\times 2N$-matrix. It is not difficult to see that the characteristic polynomial associated with this matrix is factorized as
\ben
\lambda^2(\lambda^{2N}-\alpha),
\een
with the constant $\alpha$ defined by \eqref{eq:Def-Alpha}. This implies the above proposition.
\end{proof}

As we see in Proposition~\ref{prop:DoubleBarrier}, non-zero resonances in the double barrier problem are simple. We give an example of multiple resonances in the triple barrier problem. 

\begin{proposition}
Assume that $C(x)=I_2$ holds for $x\in\{0,\pm1\}$. Then each non-zero resonance is a root of the equation
\be\label{eq:chara-poly-3}
\lambda^4-(c_{21}(0)c_{12}(1)+c_{21}(-1)c_{12}(0))\lambda^2-c_{21}(-1)c_{12}(1)\det C(0)=0,
\ee
where $c_{jk}(x)$ stands for the $(j,k)$-entry of $C(x)$. Its multiplicity coincides with that as a root. In particular, 
\ben
\pm\frac1{\sqrt2}(c_{21}(0)c_{12}(1)+c_{21}(-1)c_{12}(0))^{1/2},
\een
are non-zero resonances of multiplicity two if and only if
\be\label{eq:multicond1}
(c_{21}(0)c_{12}(1)+c_{21}(-1)c_{12}(0))^2+4c_{21}(-1)c_{12}(1)\det C(0)=0.
\ee
\end{proposition}

\begin{proof}
This proposition is proved also by a direct computation of the roots of the characteristic polynomial.
Let $J=[-1,1]=\ope{chs}(C-I_2)$. Then $\mathbbm{1}_JU\mathbbm{1}_J$ is identified with the matrix 
\ben
E_J:=
\begin{pmatrix}
0&0&c_{11}(0)&c_{21}(0)&0&0\\
0&0&0&0&0&0\\
0&0&0&0&c_{11}(1)&c_{21}(1)\\
c_{21}(-1)&c_{22}(-1)&0&0&0&0\\
0&0&0&0&0&0\\
0&0&c_{21}(0)&c_{22}(0)&0&0
\end{pmatrix},
\een
in the sense that we have 
\ben
\begin{pmatrix}\mathbbm{1}_JU\mathbbm{1}_J\psi(-1)\\\mathbbm{1}_JU\mathbbm{1}_J\psi(0)\\\mathbbm{1}_JU\mathbbm{1}_J\psi(1)\end{pmatrix}=E_J\begin{pmatrix}\psi(-1)\\\psi(0)\\\psi(1)\end{pmatrix}\qtext{for any}\psi\in\cH_{\ope{loc}}.
\een
Then the characteristic polynomial of this matrix is $\lambda^2$ times \eqref{eq:chara-poly-3}. 
\end{proof}

\begin{remark}
Let us simplify the condition \eqref{eq:multicond1}. 
The modulus of the two terms necessarily coincide when the equality holds: 
\ben
\left|c_{12}(0)\right|=\frac{2\sqrt{|c_{21}(-1)c_{12}(1)|}}{|c_{21}(-1)+c_{12}(1)(c_{21}(0)/c_{12}(0))|}.
\een
Hence, by putting $r_x:=\left|c_{12}(x)\right|=\left|c_{21}(x)\right|$, they need to satisfy
\ben
r_0\ge\frac{2\sqrt{r_{-1}r_1}}{r_{-1}+r_1}.
\een
It is required that $r_{-1}\neq r_1$ since the right hand side is equal to one when $r_{-1}=r_1$ and $r_0<1$. 

Let us restrict ourselves to the case with $c_{11}(x)=c_{22}(x)=\sqrt{1-r_x^2}$ and $c_{12}(x)=-c_{21}(x)=r_x$. Then the condition \eqref{eq:multicond1} turns into 
\ben
r_0=\frac{2\sqrt{r_{-1}r_1}}{r_{-1}+r_1}.
\een
For any $r_{-1}\neq r_1$, the right hand side is between $0$ and $1$. 
The following gives an example:
\ben
r_{-1}=\frac{3}{4},\quad r_0=\frac{12}{13},\quad r_1=\frac{1}{3}.
\een
\end{remark}

\subsection{Symmetries}
The distribution of resonances is symmetric in the following sense.
\begin{proposition}\label{eq:Periodic-Supp}
If $\lambda$ is a resonance of $U$, then $-\lambda$ is also a resonance and $m(-\lambda)=m(\lambda)$. 
In general, for any $k\in\N\setminus\{0\}$ satisfying $\ope{supp}\,c_{12}
\subset k\Z$, we have 
\be
m(e^{il\pi/k}\lambda)=m(\lambda)\quad(l=0,1,2,\ldots,2k-1)
\ee 
for any $\lambda\in\C$. Moreover, if $\pphi_\lambda$ is a (generalized) resonant state associated with $\lambda\in\ope{Res}(U)\setminus\{0\}$, then $\pphi\in\cH_{\ope{loc}}$ given by 
\be
\pphi(x):=\begin{pmatrix}e^{il\pi x/k}&0\\0&e^{-il\pi x/k}\end{pmatrix}\pphi_\lambda(x)
\ee
is that associated with $e^{il\pi x/k}\lambda$. 
\end{proposition}

\begin{proof}
Let us suppose that $\ope{supp}\,c_{12}\subset k\Z$ with $k\in\N\setminus\{0\}$. It suffices to show
\be\label{eq:Commute-zeta}
U\begin{pmatrix}e^{il\pi x/k}&0\\0&e^{-il\pi x/k}\end{pmatrix}=e^{il\pi /k}U,
\ee
for any $l\in\{0,1,\ldots,2k-1\}$. The matrix $\ope{diag}(e^{i\pi x/k},e^{-i\pi x/k})$ commutes with $C(x)$ for any $x\in\Z$. Note that the latter matrix $C(x)$ is also a diagonal matrix for $x\notin k\Z$, and the former matrix is $I_2$ for $x\in k\Z$. Hence Identity \eqref{eq:Commute-zeta} follows from
\ben
S\begin{pmatrix}e^{i\pi x/k}&0\\0&e^{-i\pi x/k}\end{pmatrix}=e^{i\pi /k}S.
\een
\end{proof}

\begin{remark}
If each $C(x)$ is a (real) orthogonal matrix, then we have
\ben
m(\lambda)=m(\bar\lambda)
\een
for any $\lambda\in\C\setminus\{0\}$, and for a (generalized) resonant state $\pphi_\lambda$ associated with $\lambda\in\ope{Res}(U)\setminus\{0\}$, 
\ben
\pphi=\overline{\pphi_\lambda}
\een 
is that associated with $\bar{\lambda}$. These facts follows from
\ben
\overline{U\psi}=U\overline{\psi}
\een
for any $\psi\in\cH_{\ope{loc}}$.
\end{remark}

We can also define \textit{incoming resonances} as poles of meromorphic family of operators $(\mathbbm{1}_J(U-\lambda)^{-1}\mathbbm{1}_J)_{\lambda\in\C}$ extended from $\left|\lambda\right|<1$ to whole $\C$. They have similar properties with (outgoing) resonances. Especially, $\lambda\in\C$ is a incoming resonance if and only if there exists an incoming map $\pphi^\inc\in\cH_{\ope{loc}}$ satisfying \eqref{eq:EigenEquation}. 
Here, we say a map $\psi\in\cH_{\ope{loc}}$ is \textit{incoming} if there exists $r>0$ such that
\ben
\psi^L(-x)=\psi^R(x)=0
\een
holds for any $x>r$. We denote the multiplicity of an incoming resonance $\lambda$ by $m^\inc(\lambda)$.

\begin{proposition}
If $c_{11}(x)=c_{22}(x)$ holds for any $x\in\Z$, we have for $\lambda\in\C\setminus\{0\}$,
\be
m(\lambda)=m^\inc(\overline{\lambda}^{-1}).
\ee
Moreover, if $\pphi_\lambda$ is a (generalized) resonant state for an outgoing resonance $\lambda\in\ope{Res}(U)\setminus\{0\}$, then 
\ben
\pphi:=S\begin{pmatrix}0&1\\1&0\end{pmatrix}\overline{\pphi_\lambda}
\een
is a (generalized) resonant state for the incoming resonance $\overline{\lambda}^{-1}$. 
\end{proposition}

\begin{proof}
The assumption $c_{11}(x)=c_{22}(x)$ implies that 
\ben
\overline{PC(x)^*}=C(x)P
,\qtext{with} P:=\begin{pmatrix}0&1\\1&0\end{pmatrix},
\een
for each $x\in\Z$. Since we also have $PS^*=SP$, 
it follows that
\ben
\overline{SPU^*\psi}
=S\overline{PC^*S^*\psi}
=SCPS^*\overline{\psi}
=USP\overline{\psi}.
\een
Hence, $(U-\lambda)\psi=f$ (equivalently $U^*\psi=\lambda^{-1}(U^*f-\psi)$) implies that
\ben
USP\overline{\psi}=SP\overline{U^*\psi}=\overline{\lambda}^{-1}(SP\overline{\psi}-USP\overline{f}),
\een
and $(U-\overline{\lambda}^{-1})(SP\overline{\psi})=-\overline{\lambda}^{-1}USP\overline{f}$.
This ends the proof.
\end{proof}

\section{Generic simplicity of resonances}\label{sec:GenSimp}
We show Theorem~\ref{thm:GenSimp} which tells us that ``most of" quantum walks satisfying Assumption has only simple resonances except the zero resonance, that is, the image of the multiplicity function $m$ on $\ope{Res}(U)\setminus\{0\}$ is a subset of $\{1\}$. In the preliminaries for the proof, we also give other characterizations of resonances using the transfer matrix (Lemma~\ref{lem:Chara-T}) or the scattering matrix (Corollary~\ref{cor:Chara-S}).

For any $k\in\N$, let $\mathfrak{Q}_k$ be the set of equivalent classes of quantum walks satisfying Assumption and $\left|\ope{chs}(C-I_2)\right|_\Z=k$, where two quantum walks $U_1=SC_1$ and $U_2=SC_2$ are said to be equivalent if the coins are translation of each other, that is, there exist $y\in\Z$ such that $C_1(x)=C_2(x-y)$ holds for all $x\in\Z$. We introduce the group topology of $\mathfrak{C}^k$ ($\mathfrak{C}:=\{C\in \ope{U}(2);\,C_{11}\neq0\}$) to $\mathfrak{Q}_k$ by the trivial bijection. We define later the group operation and the topology here of $\mathfrak{C}$ which make $\mathfrak{C}$ a Hausdorff topological group. 

\begin{theorem}\label{thm:GenSimp}
For any $k\in\N$, the set of $U$'s satisfying $m(\ope{Res}(U)\setminus\{0\})\subset\{1\}$ is dense in $\mathfrak{Q}_k$. 
\end{theorem}

We introduce the topology 
by using the transfer matrices. Let $\mathfrak{T}$ be the set of $2\times2$ matrices given by
\ben
\mathfrak{T}=\left\{e^{i\theta}T;\,T\in\ope{SL}(2),\ T_{11}=\overline{T}_{22},\ T_{21}=\overline{T}_{12},\ \theta\in[0,\pi)\right\}.
\een
Here we denote the $(j,k)$-entry of a matrix $T$ by $T_{jk}$. 
Then $\mathfrak{T}$ and $\mathfrak{C}$ are isomorphic to $\mathfrak{G}:=\left(\{(z,w)\in\C^2;|z|^2-|w|^2=1\}\times(\R/2\pi\Z)\right)/\sim$, where we define an equivalence relation $(p,q,\theta)\sim(-p,-q,\theta-\pi)$. The isomorphisms from $\mathfrak{G}$ to $\mathfrak{T}$ and to $\mathfrak{C}$ are given respectively by 
\ben
(p,q,\theta)\mapsto T_{p,q,\theta}:=
e^{i\theta}
\begin{pmatrix}
p&\bar{q}\\q&\bar{p}
\end{pmatrix}
\qtext{and}
(p,q,\theta)\mapsto C_{p,q,\theta}:=
\bar{p}^{-1}
\begin{pmatrix}
e^{i\theta}&\bar{q}\\-q&e^{-i\theta}
\end{pmatrix}.
\een
These are well-defined and one-to-one after divided by $\sim$. We consider the topology induced by 
$\mathfrak{G}$. On $\mathfrak{T}$, we consider the usual multiplication of matrices. Then the group product $*$ on $\mathfrak{C}$ is induced by that on $\mathfrak{T}$ through the isomorphism $\mathcal{M}:\mathfrak{T}\ni T_{p,q,\theta}\mapsto C_{p,q,\theta}\in\mathfrak{C}$.

\subsection{Other characterization of non-zero resonances}
We discuss on the other characterization of non-zero resonances.

\begin{lemma}\label{lem:Chara-T}
Let $\lambda_0\in\C\setminus\{0\}$, $(\kappa_0\in\C/2\pi\Z)$. Under Assumption, $\lambda_0$ is a resonance of $U$ if and only if 
\be\label{eq:Chara-T}
\begin{pmatrix}1&0\end{pmatrix}
\mathbb{T}(\lambda_0)
\begin{pmatrix}1\\0\end{pmatrix}
=0,
\ee
where we put
\ben
\mathbb{T}(\lambda)=T_\lambda(x^+
)T_\lambda(
x^+-1)\cdots T_\lambda(x^-),
\quad[x^-,x^+]:=\ope{chs}(C-I_2),
\een
and $T_\lambda(x)$ for each $x\in\Z$ is the analytic continuation of 
\be\label{eq:T-kappa}
\begin{aligned}
&T_\lambda(x)=\mathcal{M}^{-1}(\lambda^{-1}C(x))
=
e^{i\theta(x)}\begin{pmatrix}\lambda p(x)&\overline{q(x)}\\q(x)&\lambda^{-1}\overline{p(x)}\end{pmatrix},
\\
&p(x)=\frac{e^{-i\phi(x)}}{|c_{11}(x)|},\quad
q(x)=\frac{e^{-i\phi(x)}c_{21}(x)}{|c_{11}(x)|},
\\
&\theta(x)=\frac12\ope{arg}\left(\frac{c_{22}(x)}{c_{11}(x)}\right),\quad
\phi(x)=\frac{\ope{arg}(c_{11}(x)c_{22}(x))}2,
\end{aligned}
\ee
defined for $\left|\lambda\right|=1$.
Moreover, $(1,0)\mathbb{T}(\lambda)({}^t(1,0))$ is rational function with respect to $\lambda$, and
\ben
\sigma(\lambda):=\lambda^{|\ope{chs}(C-I_2)|_{\Z}}\begin{pmatrix}1&0\end{pmatrix}
\mathbb{T}(\lambda)
\begin{pmatrix}1\\0\end{pmatrix}
\een
is a polynomial of degree $(2|\ope{chs}(C-I_2)|_{\Z})$. The multiplicity as a zero of this polynomial coincides with that as a resonance.
\end{lemma}

We call each matrix $T_\lambda (x)$ or their product $\mathbb{T}(\lambda)$ transfer matrix.

\begin{remark}
The scattering matrix $S(\lambda)$ is expressed as
\ben
S(\lambda)=\begin{pmatrix}\lambda^{-(x^--1)}&0\\0&\lambda^{x^++1}\end{pmatrix}\mc{M}(\mathbb{T}(\lambda))\begin{pmatrix}\lambda^{x^+}&0\\0&\lambda^{-x^-}\end{pmatrix}
=\frac1{t_{11}(\lambda)}\begin{pmatrix}1&-t_{12}(\lambda)\\t_{21}(\lambda)&\varDelta\end{pmatrix},
\een
where $t_{jk}(\lambda)$ stands for the $(j,k)$-entry of  $\mathbb{T}(\lambda)$, $[x^-,x^+]=\ope{chs}(C-I_2)$, and $\varDelta=\det\mathbb{T}(\lambda)$ is independent of $\lambda$ (since each factor $\det T_\lambda(x)$   is independent of $\kappa$ by definition \eqref{eq:T-kappa}). 
Recall that the scattering matrix is a unitary $2\times2$-matrix which can be introduced as the linear relationship 
\ben
\left(\til{\Psi}_L^-(\lambda),\til{\Psi}_R^+(\lambda)\right)S(\lambda)=\left(\til{\Psi}_L^+(\lambda),\til{\Psi}_R^-(\lambda)\right),
\een
between the bases 
consist of Jost solutions to the equation $(U-\lambda)g=0$ characterized by
\ben
\til{\Psi}_L^\pm(\lambda,x)=\Psi_L(\lambda,x),\quad
\til{\Psi}_R^\pm(\lambda,x)=\Psi_R(\lambda,x),\qtext{for} \pm x\gg1.
\een

\end{remark}

\begin{corollary}\label{cor:Chara-S}
Suppose that $\varDelta\neq-1$ (a sufficient condition is given for example by $c_{11}(x)=c_{22}(x)$  $\forall x\in\Z$). 
Then non-zero resonances including their multiplicity coincide with poles of the meromorphic continuation of $\ope{tr}\, S(\lambda)$. In particular, we have
\ben
m(\lambda)=\frac1{2\pi i}\oint_{\lambda}\ope{tr}\left((\p_\mu S(\mu))S(\mu)^{-1}\right)d\mu.
\een
\end{corollary}
\begin{remark}
Under the same assumption with the above corollary, incoming resonances  including their multiplicity coincide with poles of the meromorphic continuation of $\ope{tr} (S(\lambda)^{-1})$.
\end{remark}

In order to use the transfer matrices, we introduce a unitary operator $Q:\cH\to\cH$ (naturally extended to $\cH_{\ope{loc}}\to\cH_{\ope{loc}}$) by
\ben
Q\psi(x)=\begin{pmatrix}1&0\\0&0\end{pmatrix}\psi(x-1)+\begin{pmatrix}0&0\\0&1\end{pmatrix}\psi(x)
=\begin{pmatrix}\psi^L(x-1)\\\psi^R(x)\end{pmatrix}.
\een
It is clear  (and known) that for any $\psi\in\cH_{\ope{loc}}$ and $\lambda\in\C\setminus\{0\}$, the equality $(U-\lambda)\psi=0$ is equivalent to that 
\be\label{eq:Transfer-equiv}
Q\psi(x+1)
=T_\lambda(x)Q\psi(x),
\ee
holds for every $x\in\Z$. 
In particular, it follows that
\be\label{eq:T-lambda-psi}
Q\psi(x^++1)=\mathbb{T}(\lambda)Q\psi(x^-),
\quad [x^-,x^+]:=\ope{chs}(C-I_2).
\ee

\begin{proof}[Proof of Lemma~\ref{lem:Chara-T}]
By definition, $\psi$ is outgoing if and only if
\be\label{eq:perp}
\begin{pmatrix}0&1\end{pmatrix}Q\psi(x^-)
\qtext{and}
\begin{pmatrix}1&0\end{pmatrix}Q\psi(x^++1).
\ee
Then the characterization \eqref{eq:Chara-T} for each non-zero resonance follows from \eqref{eq:T-lambda-psi} and \eqref{eq:perp}. 

Let $\pphi_\lambda\in\cH_{\ope{loc}}$ be defined inductively by
\ben
Q\pphi_\lambda(x)
:=T_\lambda(x)^{-1}Q\pphi_\lambda(x+1)
\een
for $x<x^-$ and
\ben
Q\pphi_\lambda(x)
:=T_\lambda(x)Q\pphi_\lambda(x-1)
\een
for $x>x^-$ with
\ben
Q\pphi_\lambda(x^-):=\begin{pmatrix}1\\0\end{pmatrix}.
\een
This is well defined as a solution to $(U-\lambda)\pphi=0$ for any $\lambda\in\C\setminus\{0\}$, and is outgoing if and only if \eqref{eq:Chara-T} holds. Since $T_\lambda(x)$ is analytic with respect to $\lambda$ away from $\lambda=0$, we have for $k\ge1$,
\ben
\p_\lambda^k(U-\lambda)\pphi_\lambda=(U-\lambda)\p_\lambda^k\pphi_\lambda-\p_\lambda^{k-1}\pphi_\lambda.
\een
By definition, $\p_\lambda^k\pphi$ is outgoing if and only if $\sigma^{(k)}(\lambda)$ vanishes. 
Moreover, a Jordan chain $(\pphi_1,\ldots,\pphi_{k+1})$ is constructed by
\be\label{eq:Const-Jor-Ch}
\pphi_{l+1}:=\frac{1}{l!}\p_\lambda^{l}\pphi_\lambda\big|_{\lambda=\lambda_0}\quad(l=0,1,\ldots,k).
\ee
This with Proposition~\ref{thm:DefRes} (\ref{enu:Jor-Chain}) implies that the multiplicity of the zero of the polynomial $\sigma(\lambda)$ at $\lambda=\lambda_0\in\C\setminus\{0\}$ is bounded by $m(\lambda_0)$. 

It suffices to show for each non-zero resonance $\lambda_0\in\ope{Res}(U)\setminus\{0\}$ that $\sigma^{(l)}(\lambda_0)=0$ for $l=0,1,\ldots,m(\lambda_0)-1$ (note that $l=0$ has been already shown). 
We prove it by a mathematical induction. Assume for an $m'\in\{0,1,\ldots,m(\lambda_0)-1\}$ that $\sigma^{(l)}(\lambda_0)=0$ holds for $l=0,1,\ldots,m'-1$. We prove $\sigma^{(m')}(\lambda_0)=0$ under this induction hypothesis. 

From the hypothesis, there exists the Jordan chain $(\pphi_1,\ldots,\pphi_{m'})$ defined by \eqref{eq:Const-Jor-Ch} up to $m'$. 
The definition is rewritten as 
\ben
Q\pphi_l(x)=\frac1{(l-1)!}\left(\p_\lambda^{l-1}\mc{T}_\lambda(x-1)\right)\big|_{\lambda=\lambda_0}e_1,
\een
where we put
\ben
\mc{T}_\lambda(x)=\left\{
\begin{aligned}
&T_\lambda(x)T_\lambda(x-1)\cdots T_\lambda(x^-)&&x>x^-,\\
&T_\lambda(x+1)^{-1}T_\lambda(x+2)^{-1}\cdots T_\lambda(x^-)&&x<x^-,
\end{aligned}\right.
\quad
e_1=\begin{pmatrix}1\\0\end{pmatrix}.
\een
From the argument above, $\sigma^{(m')}(\lambda_0)=0$ is equivalent to say that $\til{\psi}\in\cH_{\ope{loc}}$ defined by 
\be\label{eq:Aim-Cond}
Q\til{\psi}(x)=\frac1{m'!}\left(\p_\lambda^{m'}\mc{T}_\lambda(x-1)\right)\big|_{\lambda=\lambda_0}e_1
\ee
is outgoing.

According to Proposition~\ref{thm:DefRes} (\ref{enu:Jor-Chain}), there exists a generalized resonant state $\psi\in\cH_{\ope{loc}}$ such that 
\ben
(U-\lambda)^k\psi\neq0\quad(k=0,\ldots,m'),\quad(U-\lambda)^{m'+1}\psi=0.
\een
Furthermore, we can take such $\psi$ that 
\ben
(U-\lambda)\psi=\pphi_{m'},\quad
Q\psi(x^-)=0.
\een
This is due to the uniqueness of the Jordan chain (Proposition~\ref{thm:DefRes} (\ref{enu:Jor-Chain})) and to the fact that $Q\psi(x^-)$ is parallel to $Q\pphi_1(x^-)$.
Then our aim \eqref{eq:Aim-Cond} is reduced to proving $\psi=\til{\psi}$, namely the identity
\be\label{eq:Aim-Cond2}
Q\psi(x)=\frac{1}{m'!}\left(\p_\lambda^{m'}\mc{T}_\lambda(x-1)\right)\big|_{\lambda=\lambda_0}e_1.
\ee

In an almost same way as \eqref{eq:Transfer-equiv}, we see for any $u,f\in\cH_{\ope{loc}}$ that $(U-\lambda)u=f$ is equivalent to 
\be\label{eq:ind-inhomo}
Qu(x+1)=T_\lambda(x)Qu(x)-A_\lambda(x)\left[\begin{pmatrix}1&0\\0&0\end{pmatrix}Qf(x)+\begin{pmatrix}0&0\\0&1\end{pmatrix}Qf(x+1)\right],\quad\forall x\in\Z,
\ee
with
\ben
A_\lambda(x)=\frac1{c_{11}(x)\lambda}\begin{pmatrix}c_{11}(x)&0\\c_{21}(x)&-\lambda\end{pmatrix}.
\een 
By Lemma~\ref{lem:inhomo}, it follows from \eqref{eq:ind-inhomo} and the definition of $\pphi_{m'}$ that 
\be\label{eq:conc-inhomo}
Qu(x)=\mc{T}_\lambda(x-1)Qu(x^-)+\frac{1}{m'!}\left(\p_\lambda^{m'}\mc{T}_\lambda(x-1)\right)\big|_{\lambda=\lambda_0}e_1.
\ee
Then $Q\psi(x^-)=0$ implies the identity \eqref{eq:Aim-Cond2}.
\end{proof}

\begin{lemma}\label{lem:inhomo}
If $u\in\cH_{\ope{loc}}$ satisfies $(U-\lambda)u=\pphi_{m'}$, then the identity \eqref{eq:conc-inhomo} is true.
\end{lemma}

\begin{proof}
By a straightforward computation, we see
\begin{align}
&\label{eq:Lem1}
-A_\lambda(x)\left[\begin{pmatrix}1&0\\0&0\end{pmatrix}+\begin{pmatrix}0&0\\0&1\end{pmatrix}T_\lambda(x)\right]=\p_\lambda T_\lambda(x),\\
&\label{eq:Lem2}
-A_\lambda(x)\begin{pmatrix}0&0\\0&1\end{pmatrix}\p_\lambda^l T_\lambda(x)
=\frac1{l+1}\p_\lambda^{l+1} T_\lambda(x)\quad(l\ge1).
\end{align}
Then it follows that 
\be\label{eq:Lem3}
\begin{aligned}
&-A_\lambda(x)\left[\begin{pmatrix}1&0\\0&0\end{pmatrix}Q\pphi_{m'}(x)+\begin{pmatrix}0&0\\0&1\end{pmatrix}Q\pphi_{m'}(x+1)\right]
\\&
\qquad=-\frac1{(m'-1)!}A_\lambda(x)\left[\begin{pmatrix}1&0\\0&0\end{pmatrix}\left(\p_\lambda^{m'-1}\mc{T}_\lambda(x-1)\right)+\begin{pmatrix}0&0\\0&1\end{pmatrix}\left(\p_\lambda^{m'-1}\mc{T}_\lambda(x)\right)\right]e_1\big|_{\lambda=\lambda_0}\\&
\qquad=\frac1{m'!}\sum_{l=0}^{m'-1}\begin{pmatrix}m'\\l+1\end{pmatrix}(\p_\lambda^{l+1}T_\lambda(x))\left(\p_\lambda^{m'-1-l}\mc{T}_\lambda(x-1)\right)e_1\big|_{\lambda=\lambda_0}.
\end{aligned}
\ee

We prove \eqref{eq:conc-inhomo} by using the mathematical induction with respect to $x$. We here only show it for $x\ge x^-$. For $x=x^-$, the identity is just a convention $\mc{T}(x^--1)=I_2$. 
Assume that \eqref{eq:conc-inhomo} holds for an $x\ge x^-$. Then for $x+1$, we have
\ben
\begin{aligned}
Qu(x+1)=
&T_\lambda(x)\mc{T}_\lambda(x-1)Qu(x^-)+\frac1{m'!}T_\lambda(x)\left(\p_\lambda^{m'}\mc{T}_\lambda(x-1)\right)\big|_{\lambda=\lambda_0}e_1\\
&-A_\lambda(x)\left[\begin{pmatrix}1&0\\0&0\end{pmatrix}Q\pphi_{m'}(x)+\begin{pmatrix}0&0\\0&1\end{pmatrix}Q\pphi_{m'}(x+1)\right].
\end{aligned}
\een
The required formula is obtained by substituting \eqref{eq:Lem3} to this.
\end{proof}

\subsection{Proof of the generic simplicity}
The characterization due to the previous subsection shows that the resonances are zeros of a polynomial. 
The following lemma gives the instability of zeros of a holomorphic function under a small perturbation. 
\begin{lemma}[Theorem 2.26 of \cite{DyZw}]
\label{lem:Rouche}
Let $\e\mapsto f_\e(z)$ be a family of functions holomorphic in a complex disc $D(0,r_0)=\{z\in\C;\,|z|\le r_0\}$ satisfying 
\ben
f_\e(z)=z^m-\e+\ord(\e^2)+\ord(\e|z|),\quad \left|z\right|\le r_0.
\een
Then for $\e$ sufficiently small, $f_\e(z)$ has exactly $m$ simple zeros in $D(0,r_0)$.
\end{lemma}

\begin{proof}[Proof of Proposition~\refup{thm:GenSimp}]
Let $U\in\mathfrak{Q}_k$. We assume for simplicity that $\ope{chs}(C-I_2)=[1,k]$. 
Let $\{U_{\vartheta,\e};\,\vartheta\in(\R/2\pi\Z),\,0<\e\le\e_0\}$ be a family of quantum walks defined by modifying $U$ in the following way. We replace the coin matrix $C(k)$ by 
\ben
B(\vartheta,\e):=C_{p(\vartheta,\e),q(\vartheta,\e),0}*C(k),\quad
p(\vartheta,\e)=\frac{1+\e e^{i\vartheta}}{\sqrt{1+2\e\cos\vartheta}},\quad
q(\vartheta,\e)=\frac{\e e^{i\vartheta}}{\sqrt{1+2\e\cos\vartheta}}.
\een
Then $B(\vartheta,\e)\to C(k)$ as $\e\to0^+$ uniformly with respect to $\vartheta$. 
Hence, it suffices to show that for any $\e_0>0$, there exist $0<\e_1\le\e_0$ and $\vartheta_1\in\R/2\pi\Z$ such that every non-zero resonance of $U_{\vartheta_1,\e_1}$ is simple.
Put
\ben
\begin{aligned}
\sigma_k(\lambda):=\lambda^k\begin{pmatrix}1&0\end{pmatrix}\mathbb{T}(\lambda)\begin{pmatrix}1\\0\end{pmatrix}
&=\left(\begin{pmatrix}\lambda&0\end{pmatrix}T_\lambda(k)\right)\left(T_\lambda(k-1)\cdots T_\lambda(1)\begin{pmatrix}\lambda^{k-1}\\0\end{pmatrix}\right)
\\&=:e^{i\theta(k)}\begin{pmatrix}\lambda^2p(k)&\lambda\overline{q(k)}\end{pmatrix}\begin{pmatrix}\sigma_{k-1}(\lambda)\\\tau_{k-1}(\lambda)\end{pmatrix}.
\end{aligned}
\een
Then $\sigma_k$, $\sigma_{k-1}$, and $\tau_{k-1}$ are polynomials of degree $2k$, $2k-2$, and $2k-3$, respectively.
For each perturbed quantum walk $U_{\vartheta,\e}$, the transfer matrix $\mathbb{T}_{\vartheta,\e}(\lambda)$ is given by replacing $T_\lambda(k)$ by the analytic continuation of $\mc{M}^{-1}(e^{i\kappa}B(\vartheta,\e))$, and we have
\ben
\sigma_k(\lambda,\vartheta,\e):=\lambda^k\begin{pmatrix}1&0\end{pmatrix}\mathbb{T}_{\vartheta,\e}(\lambda)\begin{pmatrix}1\\0\end{pmatrix}
=\sigma_k(\lambda)+\e e^{i\theta(k)}\begin{pmatrix}\lambda^2\alpha(\vartheta)&\lambda\beta(\vartheta)\end{pmatrix}\begin{pmatrix}\sigma_{k-1}(\lambda)\\\tau_{k-1}(\lambda)\end{pmatrix}+\ord(\e^2),
\een
with
\ben
\alpha(\vartheta)=i(\sin\vartheta)p(k)+e^{i\vartheta}q(k),\quad
\beta(\vartheta)=i(\sin\vartheta)\overline{q(k)}+e^{-i\vartheta}\overline{p(k)}.
\een
Note that the vector ${}^t(\sigma_{k-1}(\lambda),\tau_{k-1}(\lambda))$ is the first column vector of the product $\lambda^{k-1}T_\lambda(k-1)\cdots T_\lambda(1)$. It does not equal ${}^t(0,0)$ for any $\lambda$ since the determinant of these matrices never vanishes. 

Let $\lambda_0\in\ope{Res}(U)\setminus\{0\}$ be a multiple resonance: $m(\lambda_0)\ge2$. 
According to Lemma~\ref{lem:Chara-T}, $\lambda_0$ is a zero of multiplicity $m(\lambda_0)$ of $\sigma_k(\lambda)$.
By the Taylor expansion near $\lambda_0$, we have
\ben
\sigma_k(\lambda,\vartheta,\e)=c(\lambda-\lambda_0)^{m(\lambda_0)}-\e\gamma(\vartheta,\lambda_0)+\ord(|\lambda-\lambda_0|^{m(\lambda_0)+1})+\ord(\e^2)+\ord(\e|\lambda-\lambda_0|),
\een
for some constant $c\neq0$ and 
\ben
\gamma(\vartheta,\lambda_0)=e^{i\theta(k)}(\lambda_0^2\alpha(\vartheta)\sigma_{k-1}(\lambda_0)+\lambda_0\beta(\vartheta)\tau_{k-1}(\lambda_0)).
\een
There are at most finite number of $\vartheta$ for each $\lambda\neq0$ such that $\gamma(\vartheta,\lambda_0)=0$. Remember that the number of non-zero resonances is at most finite, and hence there exists $\vartheta_0\in\R/2\pi\Z$ such that $\gamma(\vartheta_0,\lambda)\neq0$ for any $\lambda\in\ope{Res}(U)\setminus\{0\}$ with $m(\lambda)\ge2$.
Then Lemma~\ref{lem:Rouche} implies that there are only simple zeros of $\sigma_k(\lambda,\vartheta_0,\e)$ for any small $\e>0$ near each $\lambda_0$. This ends the proof.
\end{proof}


\section*{Acknowledgements}
K.H. is grateful to the Grant-in-Aid for JSPS Fellows (Grant No.~JP22J00430). 
H.M. was supported by the Grant-in-Aid for Young Scientists Japan Society for the Promotion of Science (Grant No.~20K14327). 
E.S. acknowledges financial supports from the Grant-in-Aid of 
Scientific Research (C) Japan Society for the Promotion of Science (Grant No.~19K03616) and Research Origin for Dressed Photon.

Kenta Higuchi, 
Graduate School of Science and Engineering,
Ehime University/ 
Bunkyocho 3, Matsuyama, Ehime, 
790-8577,  Japan

\textit{E-mail address}: higuchi.kenta.vf@ehime-u.ac.jp

Hisashi Morioka, 
Graduate School of Science and Engineering, 
Ehime University/ 
Bunkyo-cho 3, Matsuyama, Ehime, 
790-8577, Japan

\textit{E-mail address}: morioka@cs.ehime-u.ac.jp

Etsuo Segawa, 
Graduate School of Environment Information Sciences, 
Yokohama National University/ 
Hodogaya, Yokohama, 
240-8501, Japan

\textit{E-mail address}: segawa-etsuo-tb@ynu.ac.jp


\begin{thebibliography}{40}


\bibitem
{BER} V.~Baladi, J.-P.~Eckmann, and D.~Ruelle\,:
\newblock{Resonances for intermittent systems,}
\newblock{\it Nonlinearity 2 (1989),} 119--131.


\bibitem
{BST} O.~Bourget, D.~Sambou, and A.~Taarabt\,:
\newblock{On the spectral properties of non-selfadjoint discrete Schr\"odinger operator,}
\newblock{\it J. Math. Pure. Appl.,} \textbf{141} (2020), 1--49.

\bibitem
{BuZw} N.~Burq and M.~Zworski\,:
\newblock{Resonance expansions in semi-classical propagation,}
\newblock{\it Commun. Math. Phys.,} 232 (2001), 1--12.

\bibitem
{Cha} C.~M.~Chandrashekar\,:
\newblock{Two-component Dirac-like Hamiltonian for generating quantum walk on one-, two- and three-dimensional lattices,}
\newblock{\it Sci. Rep.-UK.,} 3, 2829, (2013).


\bibitem
{DyZw} S.~Dyatlov, M.~Zworski\,:
\newblock{\it Mathematical Theory of Scattering Resonances.} Graduate Studies in Mathematics, 200.
\newblock{American Mathematical Soc.,} 2019.
     
\bibitem
{FeHi} E.~Feldman, M.~Hillery\,:
\newblock{Quantum walks on graphs and quantum scattering theory,}
\newblock{\it Contemp. Math.,} 
edited by D.~Evans, J.~Holt, C.~Jones, K.~Klintworth, B.~Parshall, O.~Pfister, and H.~Ward, Contemporary Mathematics 381 (2005), pp71--96.

\bibitem
{FeyHi} R.~P.~Feynman, A.~R.~Hibbs\,:
\newblock{\it Quantum mechanics and path integrals,}
\newblock{Dover Publications, Inc., Minesota, NY, emended edition, 2010.}




\bibitem
{GZ} C.~Godsil and H.~Zhan\,:
\newblock{\it Discrete Quantum Walks on Graphs and Digraphs,}
\newblock{Cambridge University Press Vol. 484, 2023.}

 
 



\bibitem
{HiMo} K.~Higuchi and H.~Morioka\,:
\newblock{Complex translation methods and its application to resonances for quantum walks,}
\newblock{Work in progress}.


\bibitem
{IaKo} A.~Iantchenko and E. Korotyaev\,:
\newblock{Resonances for periodic Jacobi operators with finitely supported perturbations,}
\newblock{\it J. Math. Anal. Appl.,} 388  (2012), 1239--1253.

\bibitem
{JiZw}
L.~Jin and M.~Zworski\,:
\newblock{A local trace formula for Anosov flows.}
\newblock{with an appendix by H.~Naud,}
\newblock{\it Ann. Henri Poincar\'e,} 18, 1--35 (2017).

\bibitem
{KK} T.~Kato and S.~Kuroda\,:
\newblock{Theory of simple scattering and eigenfunction expansions,}
\newblock{\it Functional Analysis and Related Fields} (Proc. Conf. for M. Stone, Univ. Chicago, Chicago, III., 1968), 99--131. Springer, New York, 1970.

\bibitem
{Kl} F.~Klopp\,:
\newblock{Resonances for large one-dimensional ``ergodic" systems,}
\newblock{\it Anal. PDE,} 9(2)  (2016), 259--352.

\bibitem
{KKMS} T.~Komatsu, N.~Konno, H.~Morioka, and E.~Segawa\,:
\newblock{Generalized eigenfunctions for quantum walks via path counting approach,}
\newblock{\it Rev. Math. Phys.,} (2021), 33, 2150019.

\bibitem
{BBC} N.~Konno\,:
\newblock{Quantum Walks.} 
\newblock{In \it Quantum Potential Theory,} Springer 2008.

\bibitem
{Ku} P.~Kuklinski\,:
\newblock{Conditional probability distributions of finite absorbing quantum walks,}
\newblock{\it Phys. Rev. A,} 101 (2020) 032309.

\bibitem
{MW} K.~Manoucheri and J. Wang\,:
\newblock{\it Physical implementation of quantum walks,}
\newblock{Heidelberg, Springer Berlin, 2013.}


\bibitem
{MMOS} K.~Matsue, L.~Matsuoka, O.~Ogurisu, and E.~Segawa\,:
\newblock{Resonant-tunneling in discrete-time quantum walk.}
\newblock{\it Quantum Stud.: Math. Found.} 6, 35--44 (2019). 

\bibitem
{Mey} D.~A.~Meyer\,:
\newblock{From quantum cellular automata to quantum lattice gases,}
\newblock{\it J. Stat. Phys.,} 85, 551--574  (1996).

\bibitem
{MKO} K.~Mochizuki, D.~Kim, and H.~Obuse\,:
\newblock{Explicit definition of $\mathcal{PT}$ symmetry for nonunitary quantum walks with gain and loss,}
\newblock{\it Phys. Rev. A.,} 93, 062116 (2016).


\bibitem
{Moi} N.~Moiseyev\,:
\newblock{\it Non-Hermitian Quantum Mechanics.}
\newblock{Cambridge University Press, Cambridge, 2011.}

\bibitem
{Mo} H.~Morioka\,:
\newblock{Generalized eigenfunctions and scattering matrices for position-dependent quantum walks,}
\newblock{\it Rev. Math. Phys.,} 31 (2019), 1--37.

\bibitem
{MoSe} H.~Morioka and E.~Segawa\,:
\newblock{Detection of edge defects by embedded eigenvalues of quantum walks,}
\newblock{\it Quantum Inf. Process.,} 18 (2019), 283.

\bibitem
{NSZ} S.~Nakamura, P.~Stefanov, and M.~Zworski\,:
\newblock{Resonance expansions of propagators in the presence of potential barriers,}
\newblock{\it J. Funct. Anal.} 205 (2003), 180--205.

\bibitem
{Po} M.~ Pollicott\,:
\newblock{On the rate of mixing of Axiom A flows,}
\newblock{\it Invent. Math.,} \textbf{81}, (1985) 413--426.

\bibitem
{Por} R.~ Portugal\,:
\newblock{\it Quantum Walk and Search Algorithm,}
\newblock{2nd Ed., Springer Nature Switzerland} (2018).


\bibitem
{RST} S.~Richard, A.~Suzuki, and R.~Tiedra de Aldecoa\,:
\newblock{Quantum walks with an anisotropic coin II: scattering theory,}
\newblock{\it Lett. Math. Phys.,} \textbf{109} 1 (2019), 61--88. 

\bibitem
{Ru} D.~Ruelle\,:
\newblock{Resonances of chaotic dynamical systems,}
\newblock{\it Phys. Rev. Lett.,} 56, (1986), 405--407.

\bibitem
{ST} D.~Sambou and R.~Tiedra de Aldecoa\,:
\newblock{Quantum time delay for unitary operators: general theory,}
\newblock{\it Rev. Math. Phys.,} 31 (2019), 1950018.


\bibitem
{Su} A.~Suzuki\,:
\newblock{Asymptotic velocity of a position-dependent quantum walk,}
\newblock{\it Quantum Inf. Process.,} 15  (2016), 103--119. 


\bibitem
{TaZw} S.~H.~Tang and M.~Zworski\,:
\newblock{Resonance expansions of scattered waves,}
\newblock{\it Comm. Pure. Appl. Math.,} 53  (2000), 1305--1334. 

\bibitem
{Ti} R.~Tiedra de Aldecoa\,:
\newblock{Stationary scattering theory for unitary operators with an application to quantum walks,}
\newblock{\it J. Funct. Anal.,} 279 (2020), 108704. 

\bibitem
{Ts} M.~Tsujii\,:
\newblock{Contact Anosov flows and the FBI transform,}
\newblock{\it Ergod. Theory Dyn. Syst.,} 32 (2012), 2083--2118.


\bibitem
{Zw17} M.~Zworski\,:
\newblock{Mathematical study of scattering resonances,}
\newblock{\it Bull. Math. Sci.} 7(2017), 1--85.

\end{thebibliography}
\end{document}